\documentclass[12pt]{article}

\usepackage{amsmath,amsfonts,amscd,amssymb,slashed}
\usepackage{amsthm} %not compatible with the LMP style but needed if we do not use that
\usepackage[pdftex]{graphicx}
\usepackage{algpseudocode}
\usepackage{tikz}
\usepackage{enumerate}
\usepackage[hidelinks]{hyperref}
\usepackage{mathtools}
\usetikzlibrary{shapes,arrows}
\usetikzlibrary{matrix,arrows}
\usepackage{epstopdf}
\usepackage{slashed}
\usepackage{xcolor,import}
\usepackage{transparent}
\usepackage[all]{xy}

\usepackage{geometry}
\makeatletter
\@addtoreset{equation}{section}
\def\section{\@startsection {section}{1}{\z@}{-2.5ex plus -1ex minus
 -.2ex}{1.3ex plus .2ex}{\large\bf}}
\def\subsection{\@startsection{subsection}{2}{\z@}{-2.25ex plus%
 -1ex minus -.2ex}{0.5ex plus .2ex}{\bf}}

\newcommand{\bee}{\begin{equation}}
\newcommand{\eee}{\end{equation}}
\newcommand{\R}{\mathbb{R}}
\newcommand{\N}{\mathbb{N}}
\newcommand{\C}{\mathbb{C}}
\newcommand{\Z}{\mathbb{Z}}
\newcommand{\A}{\mathcal{A}}

\def\Dslash{\slashed{D}}
\def\St{\mathrm{St}}

\newtheorem{thm}{Theorem}[section]

\newtheorem{cor}[thm]{Corollary} 
\newtheorem{dfn}[thm]{Definition}
\newtheorem{lem}[thm]{Lemma}

\theoremstyle{definition}

\newcommand{\rs}{\lambda}
\tikzstyle{vertex} = [circle, draw, fill=blue!20, scale=1,auto=left]
\tikzstyle{vert} = [circle, draw, fill=blue!20, scale=.8,auto=left]
\tikzstyle{line} = [draw]

\begin{document}

\begin{flushright}
EMPG-17-06
\end{flushright}
\vskip 10pt
\baselineskip 28pt

%LMP title
%\title{Magnetic Zero-Modes, Vortices and Cartan Geometry}
%\author{Calum Ross and Bernd J.~Schroers}
%\institute{Maxwell Institute for Mathematical Sciences and
%Department of Mathematics,
%\\Heriot-Watt University,
%Edinburgh EH14 4AS, UK. \\
%{\tt cdr1@hw.ac.uk} and {\tt b.j.schroers@hw.ac.uk}}
%\maketitle

\begin{center}
{\Large \bf Magnetic Zero-Modes, Vortices and Cartan Connections}

\baselineskip 18pt

\vspace{0.4 cm}

{\bf Calum Ross and Bernd J.~Schroers}\\
\vspace{0.2 cm}
Maxwell Institute for Mathematical Sciences and
Department of Mathematics,
\\Heriot-Watt University,
Edinburgh EH14 4AS, UK. \\
{\tt cdr1@hw.ac.uk} and {\tt b.j.schroers@hw.ac.uk} \\

\vspace{0.4cm}

{ November  2017} 
\end{center}

\baselineskip 16pt
\parskip 4 pt
\parindent 10pt
\begin{abstract}
\noindent 
We exhibit a close relation between vortex  configurations on the 2-sphere and   magnetic zero-modes of the Dirac operator on $\R^3$  which obey an additional non-linear equation. We show that both    are best understood in terms of  the geometry induced on the 3-sphere  via pull-back  of the round geometry with  bundle maps  of the  Hopf fibration.  We use this viewpoint to deduce a  manifestly smooth formula for square-integrable  magnetic zero-modes in terms of two homogeneous polynomials in two complex variables. 
\end{abstract}

%\keywords{Magnetic zero modes, Dirac operator, vortex equations, Cartan geometry}
%
%\subclass{53Z05,58Z05}

\section{Introduction}
The goal of this paper is to explain and exploit   a  link between  magnetic zero-modes of the Dirac operator on Euclidean 3-space and  vortices on the 2-sphere via a particular family of geometries on the 3-sphere. 

The 3-sphere  geometries are obtained from the standard round geometry via pull-back with  a family of maps $S^3\rightarrow S^3$  which are bundle maps  of the Hopf fibration  and cover holomorphic maps $S^2\rightarrow S^2$.  The bundle  maps are given in terms of  two complex polynomials, and one  consequence of our analysis is a  manifestly smooth and square-integrable expression both for  the magnetic zero-modes and the vortex configurations in terms of these polynomials.  Another is an interpretation of vortices on $S^2$ in terms of Cartan geometry.   In the  remainder of this introduction we sketch the context for our results.

The problem of determining magnetic zero-modes of the Dirac operator in Euclidean 3-space  was first posed and addressed in an influential  paper by Loss and Yau \cite{LY} in 1986.  Motivated by  questions about the stability of atoms, the authors were interested in finding spinors $\Psi$ and  magnetic gauge potentials $A$  on $\R^3$ such that $\Psi$ is a zero-mode of the  (static) Dirac operator minimally coupled to $A$, and  both the associated   magnetic field and the spinor are  square-integrable. In this paper,  we call pairs of spinors $\Psi$ and magnetic gauge potentials $A$  satisfying this condition magnetic zero-modes. 

Loss and Yau gave explicit expressions for one family  of magnetic zero-modes, which we call linear  in the following, and derived a formula which determines a gauge field for a given spinor field such that the pair form a magnetic zero-mode. This formula is singular where the spinor field vanishes, but it was nonetheless  used by Adam, Muratori  and Nash (AMN) in a series of papers \cite{AMN1,AMN2,AMN3}  to obtain magnetic zero-modes   which   satisfy an additional non-linear equation, and which we call vortex  zero-modes in this  paper. AMN  observed  that their solutions can be expressed in terms of solutions of the Liouville equation on $S^2$, and  addressed  the  singularities in the resulting formulae.  In \cite{AMN1},  they also pointed out  that the coupled Dirac and non-linear equation can be obtained as the dimensional reduction of a perturbed   Seiberg-Witten equation on $\R^4$ with a crucial sign flip (the resulting equation is often called the Freund equation).

In 2000, Erd\"os and Solovej pointed out that  the geometry underlying the existence of   magnetic zero-modes  is  the conformal equivalence of $\R^3$ and $S^3\setminus \{ \text{point}\}$ and  the Hopf fibration of $S^3$ over $S^2$ \cite{ES}.   This was used  in \cite{DM,Min}  to show that the linear magnetic zero-modes found by Loss and Yau can be obtained  directly by pulling  eigenmodes (of any energy) of the Dirac operator on $S^3$ back to $\R^3$.  

One motivation for this paper was to find a similarly geometrical but also explicit understanding of the vortex zero-modes, i.e., to use the geometrical insight of Erd\"os and Solovej  for a better understanding and improvement of the formulae derived by AMN. A second motivation  was to explore links to  a vortex equation  for a scalar Higgs field and an abelian gauge field on $S^2$,  recently  proposed  by  Popov. The existence of such links is suggested by the appearance  of the same data  in Popov vortex solutions and  the AMN expressions for magnetic zero-modes; it is  the reason why we call the latter  vortex   zero-modes.

The Popov vortex equations  were   obtained  in \cite{Popov}  as the reduction by $SU(1,1)$ symmetry  of the self-duality equations for  $SU(1,1)$ Yang-Mills theory on the product of the 2-sphere with  hyperbolic 2-space.  In  \cite{Manton1},   Manton pointed out that the Popov vortex equations can be solved in terms of rational maps $S^2\rightarrow S^2$.  His solution turns out to be a particularly simple illustration  of an interesting  subsequent observation by Baptista \cite{Baptista} that Bogmol'nyi vortex equations on a  K\"ahler surface can be interpreted as  degenerate Hermitian metrics. 

In the terminology of Baptista's paper, Manton showed that Popov vortices encode the geometry of the pull-back of the round metric on $S^2$ with a rational map. If the rational map has degree  $n$, the metric  necessarily has  conical singularities  at the $2n-2$ ramification points, which are also the zeros of the vortex Higgs field. 

Here, we  lift this picture  from $S^2$ to $S^3$.  This is geometrically natural for  Popov vortices, since they live on a $U(1)$ bundle over $S^2$ whose total space is the Lens space $S^3/\Z_{2n-2}$. Manton's rational map characterising the vortex lifts to a bundle  map $S^3\rightarrow S^3$, and the pull-back of the round metric on $S^3$ with this bundle map defines a metric which,  generalising Baptista's viewpoint, encodes a vortex configuration on $S^3$. We then show that such a vortex configuration defines a  magnetic zero-mode  of the Dirac operator on $S^3$. Using conformal equivalence we  obtain the advertised smooth and manifestly square-integrable  expression for vortex zero-modes on $\R^3$ and, at the same time,  establish the expected link to Popov vortices.

The fact that $S^3$ is the Lie group $SU(2)$ allows one to encode the round geometry of $S^3$ in the Maurer-Cartan form $h^{-1}dh$.  In Cartan geometry, the same form also encodes the geometry of the round geometry of $S^2$ by combining the orthonormal frame field with the spin connection  1-form of $S^2$. 
Since all the geometries we discuss in this paper are pulled back from the round geometries  of $S^2$ and $S^3$, it is  not surprising that many of our results can be stated succinctly in terms of the pull-back of the Maurer-Cartan form via the bundle map $S^3\rightarrow S^3$. In fact, the flatness condition of the $su(2)$ gauge potentials defined by these pull-backs turns out to be equivalent to  our vortex equations on $S^3$ and to  the  Popov vortex equations on $S^2$. This adds a further, non-abelian interpretation of  the  vortex  zero-modes.  It also provides an intriguing link with the self-duality equations for $SU(1,1)$ gauge fields from which the Popov equations arose.

In this paper we are interested in the geometry linking magnetic zero-modes and vortices, but also in manifestly smooth expressions for both. The paper contains a number of explicit calculations and formulae, and we therefore need to lay out  conventions and coordinates  in some detail at the beginning.  To help the reader keep sight  of the bigger picture, we have  also produced a summary  diagram of the geometries and the maps between them in Fig.~\ref{summary}. Although  the picture is part of our  final summary section, the reader may find it helpful to refer to it now or while reading the paper.

The paper is organised as follows. 
In Sect.~2 we collect our conventions for parametrising $SU(2)$ both as a Lie group and a round 3-sphere,  give  the stereographic and gnomonic projection from $S^3$ and $\R^3$ in these coordinates, and explain how both enter a simple formula for  relating  orthonormal frames on  $S^3$ and $\R^3$ and for 
mapping zero-modes of  magnetic  Dirac operators on $S^3$  to zero-modes of  magnetic  Dirac operators on $\R^3$. While the conformal covariance of the  kernel of the Dirac operator is, of course, well-known, we are not aware of a treatment  which emphasises the role of the  gnomonic projection in the way we do. 
We illustrate our discussion  by constructing the linear magnetic zero-modes  on $\R^3$ from general eigenmodes on $S^3$ in our conventions.  

Sect.~3 contains our definition of vortex configurations on $S^3$ and some  of our main results: the equivalence between vortex configurations on $S^3$  and flat $su(2)$ gauge potentials, an expression for both in terms of  bundle maps $S^3\rightarrow S^3$, and the construction of magnetic zero-modes on both $S^3$ and $\R^3$  from vortex configurations. The allowed bundle maps can be expressed in terms of  two polynomials, thus leading to the promised formulae for magnetic zero-modes. The section ends with a brief discussion of how linear and vortex zero-modes can be combined to form new magnetic zero-modes. 

In Sect.~4 we review the definition of Popov vortices and  show that our vortex configurations on $S^3$ are equivariant descriptions of  them. We explain the relation between our bundle map 
$S^3\rightarrow S^3$ and the rational map $S^2\rightarrow S^2$ used by Manton for solving the Popov equations, and  interpret the pull-back to $S^2$  of the flat $su(2)$ gauge potential on $S^3$ in the language of Cartan geometry. 

Finally,  Sect.~5 contains a summary in the form of a diagram in Fig.~\ref{summary}, and an outlook onto  open questions.

\section{Magnetic Dirac operators on $S^3$ and on  $\R^3$}
\subsection{Conventions for $SU(2)$ and the Hopf map}
\label{conventions}
We  adopt the conventions of \cite{JS,JS2}, and    use  the $su(2)$  generators 
\bee
\label{liegen}
t_j=-\frac{i}{2}\tau_j, \; j=1,2,3,
\eee
where the $\tau_j$ are the Pauli matrices,  with  commutators
$[t_i,t_j]=\epsilon_{ijk}t_k$. Often we will work in terms of 
\bee
t_+=t_1+it_2, \quad  t_-=t_1-it_2
\eee
with commutators
\bee
[t_3,t_+]=-i t_+, \quad [t_3,t_-]=it_-, \quad [t_+,t_-]=-2it_3.
\eee
We parametrise an $SU(2)$ matrix $h$ in terms of a pair of complex numbers  $(z_1,z_2)$  via 
\bee
\label{complexp}
h= \begin{pmatrix}
z_1  & -\bar{z}_2\\
z_2 &
\phantom{-} \bar{z}_1 \end{pmatrix},
\eee
with the constraint $ |z_1|^2 + |z_2|^2= 1$ understood.
The (real) left-invariant 1-forms are defined via
\bee
h^{-1}dh = \sigma_1 t_1 + \sigma_2t_2 + \sigma_3 t_3,
\eee
and satisfy $
d\sigma_1 = - \sigma_2 \wedge \sigma_3$ and the cyclic permutations of this equation. 
Defining 
\bee
\sigma= \sigma_1+i\sigma_2, \qquad \bar{\sigma} = \sigma_1-i\sigma_2
\eee
we also note that 
\bee
\label{usefulsigma}
d\sigma = i\sigma \wedge \sigma_3, \quad d\sigma_3= \frac i 2 \bar{\sigma}\wedge \sigma,
\eee
and have the following expressions in  terms of complex coordinates:
\bee
\sigma= 2i (z_1dz_2 -z_2 dz_1), \quad 
\sigma_3= 2i (\bar{z}_1dz_1 + \bar{z}_2 dz_2).
\eee

The dual  vector fields  $X_j$, $j=1,2,3$,  generate the right-action $h\mapsto ht_j$ \cite{JS}. Their commutators are 
\begin{equation}
[X_i,X_j]=\epsilon_{ijk}X_k,
\end{equation}
so that,  in terms of
\begin{align}
X_+ = X_1 +iX_2, \qquad X_-=X_1 -iX_2,
\end{align}
we have 
\begin{equation}
[X_+,X_- ]=-2iX_3, \qquad [iX_3,X_\pm]=\pm X_\pm.
\end{equation}
In terms of complex coordinates, the right-generators are 
\begin{align}
X_+& =i (z_1\bar{\partial}_2 -z_2 \bar{\partial}_1), \nonumber  \\
X_-& =i (\bar{z}_2\partial_1 -\bar{z}_1 \partial_2), \nonumber \\
X_{3\;}& =\frac  i  2
 ( \bar{z}_1\bar{\partial}_1 +\bar{z}_2 \bar{\partial}_2 - z_1 \partial_1 -z_2 \partial_2).
\end{align}
The corresponding generators of the left-action $h\mapsto -t_jh$ are  denoted $Z_\pm$ and $Z_3$;  expressions in terms of our complex coordinates are given in \cite{JS}.
We also note that the Laplace operator on $S^3$  (with radius $2$) can be written as 
\bee
\label{S2S3}
\Delta_{S^3}=X_1^2 +X_2^2 +X_3^2= X_+X_- +iX_3 +X_3^2 = X_-X_+ - iX_3 +X_3^2, 
\eee
with an analogous expression in terms of $Z_j$.
Finally, we have the the pairings
\bee
\label{pairing}
\bar{\sigma}(X_+)= \sigma(X_-)=2, \quad \sigma_3(X_3)=1,
\eee
with all other pairings vanishing. 

In this paper we   parametrise  the 2-sphere via stereographic projection in terms of a coordinate $z\in \C$. We work in coordinates provided by stereographic projection from the south pole, referring to \cite{JS} for details and  the coordinate changes required to cover the south pole itself by projecting from the north pole. In terms of our complex coordinates \eqref{complexp} for $S^3$,  the Hopf map is  
\bee
\label{Hopf}
\pi: S^3 \rightarrow S^2, \quad h \mapsto z= \frac{z_2}{z_1}. 
\eee
Then
\bee
\label{locsec}
s: \C \rightarrow SU(2), \qquad z\mapsto \frac{1}{\sqrt{1+|z|^2}}\begin{pmatrix} 1  & -\bar{z} \\ z & \phantom{-}1 \end{pmatrix}
\eee
is a local section of the Hopf bundle. We will use it in this paper to switch between the equivariant description of sections of associated line bundles to local expressions for such sections.  
Defining spaces  of equivariant functions
\bee
\label{equidef}
C^\infty(S^3,\C)_N= \{F: S^3 \rightarrow \C|2iX_3F=NF\}, \quad N\in \N^0,
\eee
and writing $H$ for the hyperplane  bundle over $S^2$,
 one has the following commutative diagram \cite{JS,JS2}:
\bee
\label{equicom}
\begin{CD}
C^\infty(S^3,\C)_N  @>X_+>> C^\infty(S^3,\C)_{N+2} \\
@VV s* V   @VV s* V \\
 C^\infty(H^N) @> i((1+|z|^2)\bar \partial_z + \frac N 2 z)>>  C^\infty(H^{N+2}).
\end{CD}
\eee

\subsection{Stereographic projection and frames}
In the remainder of this paper, we  consider a 2-sphere of radius $\rs$ and a 3-sphere of radius $\ell$. Then 
\bee
\label{ONframe}
 \left(\frac \ell 2 \sigma_1, \frac \ell 2 \sigma_2, \frac \ell 2 \sigma_3\right)
\eee
provides a convenient orthonormal frame for $S^3$, and the metric is 
\bee
\label{S3met}
ds^2= \frac{\ell^2}{4}(\sigma_1^2 +\sigma_2^2 +\sigma_3^2). 
\eee

In order to make contact with the usual orientation on $\R^3$ after stereographic projection, we define the orientation on $S^3$ in terms of the volume element
\bee
\label{orient}
\Omega= \frac{\ell^3}{8}\sigma_2\wedge \sigma_1\wedge \sigma_3.
\eee

We  write elements of $\R^3$ as $
\vec{x} = (x_1,x_2,x_3)^t$, denote their length by 
$r=|\vec{x}|$, 
and assume that 
\bee
\label{r3frame}
\left(dx_1, dx_2, dx_3\right)
\eee
is an oriented orthonormal frame, so the metric  and volume element are 
\bee 
\label{euclid}
ds^2 =dx_1^2 +dx_2^2 +dx_3^2, \qquad 
dx_1\wedge dx_2 \wedge dx_3.
\eee

 Thinking of the 3-sphere of radius $\ell$ embedded in $\R^4$ with coordinates  $(y_1,y_2,y_3,y_4)$, the stereographic projection from the south pole  onto  $\R^3$ is 
\bee
\St: S^3\setminus\{\mathrm{south \;pole}\} \rightarrow \R^3, \quad (y_1,y_2,y_3,y_4)\mapsto (x_1,x_2,x_3) =
 \left(\frac{\ell y_{1}}{\ell+y_{4}},\frac{\ell y_{2}}{\ell+y_{4}},\frac{\ell y_{3}}{\ell+y_{4}}\right),
\eee
with inverse
\bee
\St^{-1}:\R^3 \rightarrow S^3, \quad (x_1,x_2,x_3)\mapsto (y_1,y_2,y_3,y_4) = \frac{\ell}{\ell^2+r^2} (2x_1\ell,2x_2\ell, 2x_3\ell, \ell^2-r^2).
\eee
Writing   $ \vec{\tau} =(\tau_1,\tau_2,\tau_3)$
for the vector whose components are the Pauli matrices, and identifying $(y_1,y_2,y_3,y_4) \in S^3$ with the unitary matrix $ (y_4\mathbb{I} + i y_1\tau_1 +  i y_2\tau_2+  i y_3\tau_3)/\ell$, the inverse stereographic projection is, up to scale, the map
\begin{align}
\label{hdef}
H:\R^3  &\rightarrow SU(2),  \nonumber \\
 \vec{x} &\mapsto\frac  {\ell^2-r^2}  {\ell^2 +r^2}  \mathbb{I}+\frac { 2i\ell}   {\ell^2+r^ 2} 
\vec{x} \cdot \vec{\tau} = \frac   {1 } {\ell^2 +r^2} \begin{pmatrix}   {\ell^2-r^2} +2i\ell x_3 & 2i\ell(x_1-ix_2) \\ 2i\ell(x_1+ix_2)
&  {\ell^2-r^2} -2i\ell x_3 \end{pmatrix},
\end{align}
so that  the Hopf projection in stereographic coordinates is 
\bee
\pi\circ H: (x_1,x_2,x_3) \mapsto  \frac{2\ell(x_1+ix_2)}{2x_3\ell + i (r^2-\ell^2)}.
\eee

In the following we will often need to pull back functions, spinors or forms on the 3-sphere  of radius $\ell$ to $\R^3$ with the inverse stereographic projection. To simplify notation we will write the pull-back in terms of $H$ rather than $\St^{-1}$, even though the two maps strictly speaking take values on 3-spheres of different radii.

A recurring theme in this paper is the interplay between the stereographic projection and the gnomonic projection, often used in cartography, which  maps great circles to straight lines. The inverse of the gnomonic projection is the map
\bee
G :\R^3\rightarrow SU(2), \qquad \vec{x} \mapsto \frac{\ell\, \mathbb{I}+ i\vec{x}\cdot \vec{\tau}}{\sqrt{\ell^2 + r^2}},
\eee
whose image  satisfies $G(\vec{x})^2=H(\vec{x})$.
This relation follows immediately from the explicit forms of the maps, but it
 also follows from the geometric meaning of $G$ and $H$, which is illustrated 
in Fig.~\ref{ghrelation}. 
 For fixed $\vec{x}$,  $G(\vec{x})$ and $H(\vec{x})$ are rotations about the same axis, and  it follows from elementary geometry that $G$ rotates by twice the rotation angle of $H$. As an aside we note that describing spherical geometry in terms of $\R^3$ via pull-back with $H$ and $G$ is analogous to describing hyperbolic geometry in terms of, respectively, the Poincar\'e and the Beltrami-Klein models.

\begin{figure}[ht]
\centering
\includegraphics[width=10truecm]{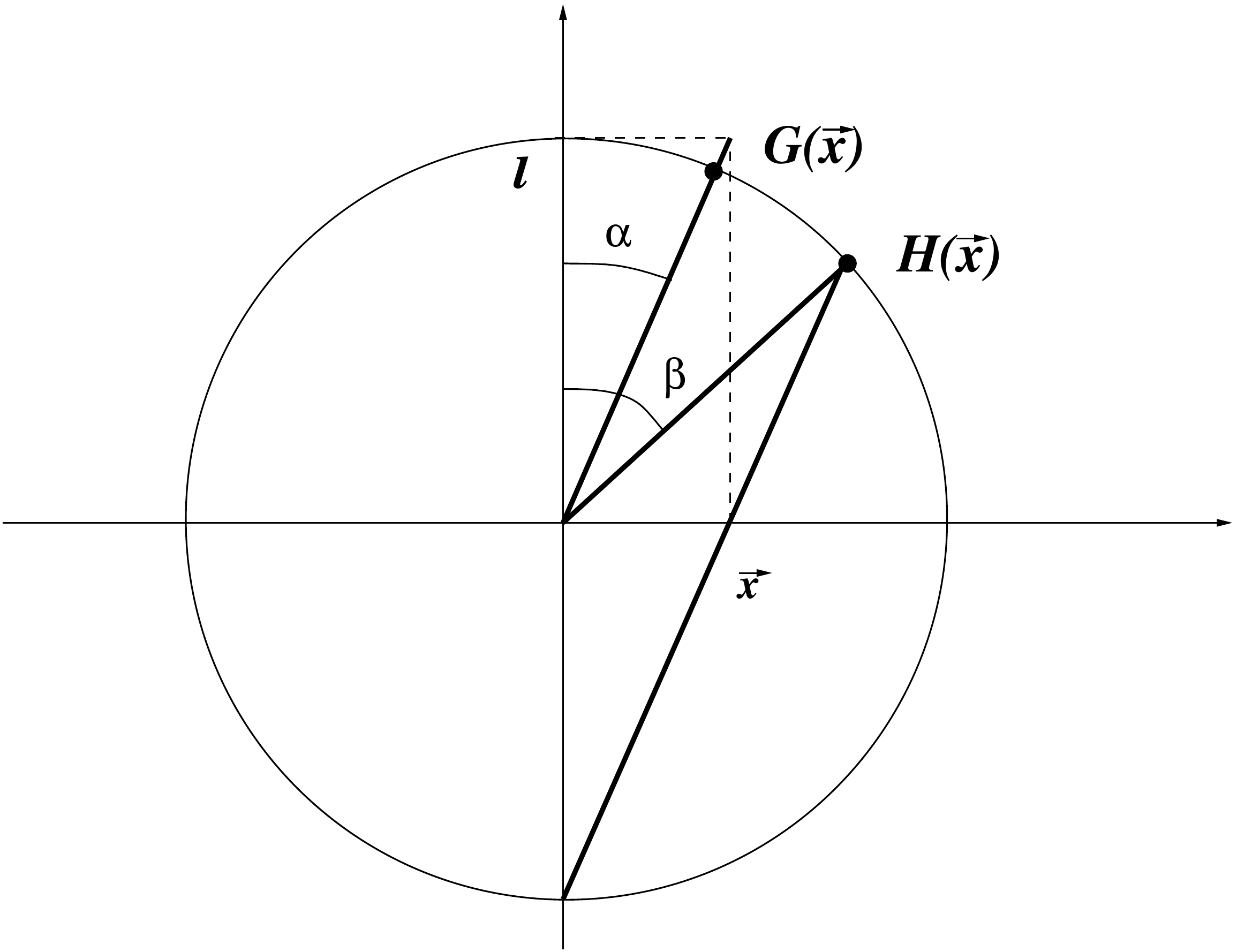}
 \caption{Geometry of the  relation between the gnomonic and stereographic  projections $G,H: \R^3\rightarrow S^3$ defined in the main text. The elementary relation $\beta= 2\alpha$ implies $H(\vec{x})= G^2(\vec{x})$.}
 \label{ghrelation}
\end{figure}

The maps $G$ and $H$ can be used to pull back the left-invariant 1-forms $\sigma_j$, $j=1,2,3$,  on the 3-sphere. The relation of the resulting frames on $\R^3$ to each other and to the standard frame \eqref{r3frame} is interesting, and important for the remainder of this paper. We therefore collect the relevant results here. 
Defining the scale function
\bee
\label{Omdef} 
\Omega=\frac{\ell^2 +r^2}{4\ell},
\eee
we introduce 1-forms $\theta_j$, $j=1,2,3$, on $\R^3$ via
\bee
\label{pullbacksigma}
H^{-1} dH =\frac{i}{2\Omega} \vec{\theta}\cdot\vec{\tau},   \quad \text{or} \quad 
\theta_j = - \Omega H^*\sigma_j.
\eee
Then we find 
\begin{align}
\label{rotatingsigma}
 \vec{\theta} \cdot \vec{\tau}  = \frac{1}{\ell^2 + r^2} \left(2(\vec{x}\cdot \vec{\tau})( \vec{x}\cdot d  \vec{x}) + (\ell^2 -r^2) \vec{\tau}\cdot d\vec{x} + 2\ell (\vec{x}\times d\vec{x}) \cdot \vec{\tau} \right)  = G^{-1} ( d \vec{x}\cdot \vec{\tau}) G.
 \end{align}
 In other words,  pulling back the 1-forms  $\sigma_1,\sigma_2$ and $\sigma_3$ via $H$ gives a frame which is related to the standard frame 
\eqref{r3frame} by rotation with $G$ (acting in its adjoint representation), a reflection in the origin  and rescaling by $\Omega$.  For later use we note 
\bee
\label{theta3exp}
\theta_3= \frac{1}{\ell^{2}+r^2}\left(
2(x_1x_3-\ell x_2)dx_1 +  2(x_2x_3+\ell x_1)dx_2 +  (\ell^{2}-r^{2}+2x_3^{2})dx_3\right).
\eee

\begin{lem}
The pull-backs $G^{-1}dG$ and $H^{-1}dH$ of the Maurer-Cartan form on $SU(2)$ are related via
\bee
 \label{HG} 
 H^{-1}dH  = G^{-1}dG + G^{-1}( G^{-1} dG)G ,
\eee
and the inverse relation can be expressed as
\bee
\label{GH}
G^{-1} dG = \frac 1 2  H^{-1}dH  + \star ( d\Omega\wedge H^{-1} dH). 
\eee
\end{lem}
\begin{proof}
The formula \eqref{HG} is an immediate consequence of $H=G^2$. To show the inverse relation, we use \eqref{HG} to write \eqref{GH} as 
\bee
\label{GHrewrite}
\star ( 2d\Omega\wedge (dG G^{-1}  +  G^{-1} dG))=dG G^{-1} - G^{-1} dG.
\eee
Computing 
\bee
 G^{-1} dG= \frac{i\ell d\vec{x} \cdot \vec{\tau} + i \vec{x}\times d\vec{x}\cdot \vec{\tau}}{\ell^2 +r^2},
\eee
we have 
\bee
\label{Gformulae}
dG G^{-1}  +  G^{-1} dG=  \frac{2i\ell d\vec{x} \cdot \vec{\tau}}{\ell^2 +r^2}, \quad dG G^{-1}  -  G^{-1} dG =
-\frac{2 i \vec{x}\times d\vec{x}\cdot \vec{\tau}}{\ell^2 +r^2}.
\eee
With $   2 d\Omega= \vec{x}\cdot d\vec{x}/\ell $  and 
\bee
\star (\vec{x}\cdot d\vec{x} \wedge d\vec{x}\cdot \vec{\tau})= d\vec{x}\times \vec{x} \cdot\vec{ \tau},
\eee
one deduces  \eqref{GHrewrite} and hence \eqref{GH}. 
\end{proof} 

\subsection{Magnetic Dirac operators and their zero-modes}
In the frame \eqref{ONframe} for a 3-sphere of radius $\ell$,   a global gauge potential for the  spin connection is 
\bee
\Gamma_{S^3} = \frac 12 h^{-1}dh,
\eee
Thus, the 
Dirac operator associated to the frame \eqref{ONframe} is 
\begin{align}
\label{S3Diracungauged}
\Dslash_{S^3} &=\frac{2i}{\ell} \left( \tau_1 X_1 + \tau_2 X_2 +  \tau_3 X_3 \right)  +\frac{3}{2\ell}\nonumber \\
& = \frac{2i}{\ell} \begin{pmatrix} X_3  & X_-  \\   X_+ & -X_3 \end{pmatrix} +\frac{3}{2\ell}.
\end{align}
Minimal coupling to  an abelian gauge potential  $A =A_1 \sigma_1 + A_2 \sigma_2 + A_3\sigma_3$  gives
\begin{align}
\label{S3Dirac}
\Dslash_{S^3,A} &=\frac{2i}{\ell} \left( \tau_1 (X_1 + iA_1) + \tau_2 (X_2 +iA_2)+  \tau_3 (X_3+iA_3) \right)  +\frac{3}{2\ell}\nonumber \\
& = \frac{2i}{\ell} \begin{pmatrix} (X_3 +iA_3)  & (X_- +iA_-)  \\   (X_+ +iA_+) & -(X_3+iA_3) \end{pmatrix} +\frac{3}{2\ell}. 
\end{align}

The Dirac operator on $\R^3$   associated to the frame \eqref{r3frame} and minimally coupled to an abelian gauge potential $A= A_1dx_1 + A_2 dx _2 + A_3dx_3$ is 
\bee
\label{flatDirac}
\Dslash_{\R^3, A}  = i\tau_j(\partial_j + i A_j), 
\eee
where  we used $\partial_1= \partial/\partial x_1$ etc. 

As we saw in the previous section,  the frame \eqref{ONframe}  pulled back to $\R^3$ via  $H$ and
the flat frame \eqref{r3frame} are related by a rotation with $G$,  a reflection and rescaling by \eqref{Omdef}. 
This implies a simple relation between the zero-modes of the Dirac operators on $S^3$ and $\R^3$.

\begin{lem}
\label{conformal}
If $\Psi: S^3\rightarrow \C^2  $ is a zero-mode of the Dirac operator \eqref{S3Dirac} on $S^3$  coupled to the $U(1)$ gauge field  $A$, then 
\bee
\label{r3spinor}
\Psi_H = G\Omega^{-1} H^*\Psi
\eee
is a zero-mode of the  Dirac operator $\Dslash_{\R^3, H^*A} $ on   Euclidean 3-space  coupled to the connection $H^*A$.
\end{lem}

\begin{proof}
Since the spin connection in the frame \eqref{r3frame} is manifestly zero, it follows from the equivariance of the Dirac operator under scaling and frame rotations that  pulling back zero-modes of the Dirac operator on $S^3$ to $\R^3$ and applying the transformation $G\Omega^{-1}$ gives zero-modes of the Dirac operator on $\R^3$ in the flat frame \eqref{r3frame}. It is instructive to check this explicitly.
The pull-back of the spin connection  is 
\bee
H^*\Gamma_{S^3}  = \frac 12 H^{-1}dH.
\eee
Thus, using \eqref{HG}, 
 \begin{align}
d + \frac 1 2 H^{-1} dH =\Omega G^{-1}\left(d+ \frac 12\left(  G  dG^{-1} +  G^{-1} dG\right) + \Omega^{-1} d\Omega     \right)\Omega^{-1} G.
\end{align}
Combining  \eqref{Gformulae} and 
\bee
\Omega^{-1}d\Omega=\frac{2\vec{x}\cdot d\vec{x}}{\ell^{2}+r^{2}},
\eee
one checks that 
\bee
\tau^{j}\iota_{\partial_j}\left(\frac 12\left(  G  dG^{-1} +  G^{-1} dG\right) + \Omega^{-1} d\Omega\right) =-\frac{2\vec{x}\cdot \vec{\tau}}{\ell^{2}+r^{2}}+\frac{2\vec{x}\cdot \vec{\tau}}{\ell^{2}+r^{2}}=0.
\eee
Combining these results, one checks that  the pull-back of the  Dirac operator  on $S^3$ coupled to the spin connection and the  abelian connection $A$ in the frame \eqref{r3frame} is 
\begin{align}
\frac \ell 2 H^*\Dslash_{S^3,A} & = \tau^j\iota_{H^*X_j} (d + \frac 1 2 H^{-1} dH +iH^*A) \nonumber \\
&= -\Omega^{-1} G^{-1}\tau^{j}\iota_{\partial_j}G(d + \frac 1 2 H^{-1} dH +iH^*A)\nonumber \\
&= -G^{-1}\tau^{j}\iota_{\partial_j}
 (d+ iH^*A)\Omega^{-1} G,
\end{align}
which implies the claimed relation between zero-modes of $\Dslash_{S^3,A}$ and $\Dslash_{\R^3,H^*A}$.
Note that the components   of the  pull-back abelian connection relative to the frame  \eqref{r3frame} are related to the  components  in the expansion $A=\vec{A}\cdot \vec{\sigma}$  via
 \bee
 H^*A = \vec{A}_H\cdot d\vec{x}, \qquad \vec{A}_H\cdot \vec{\tau} =-\frac{1}{\Omega} G^{-1}  H^{*}\vec{A}\cdot \vec{\tau}G.
 \eee
 \end{proof}
 
This lemma can be used to construct magnetic zero-modes on $\R^3$  from magnetic zero-modes on $S^3$. For the family constructed explicitly by Loss and Yau in \cite{LY}, which we call linear in this paper, this was observed in \cite{DM} and elaborated in \cite{Min}, where this family was obtained from eigenmodes of the Dirac operator on $S^3$. The corresponding argument for the  family of vortex zero-modes  is one of the main  results of our Sect.~\ref{magfields}. 

 We review the  construction of the linear zero-modes  very briefly  here because we will need them later in this paper, expressed in our conventions for parametrising $S^3$ in terms of two complex variables. We define the functions
\begin{equation}
\label{Y}
Y^j_{sm}(z_1,z_2) =C_{jms}
\sum_k\frac{  (-1)^{-k} }{(j+m-k)!k!(j-s -k)!(s -m + k)!}z_1^{s-m+k}z_2^{j+m-k}\bar{z}_1^k\bar{z}_2^{j-s-k},
\end{equation}
where  $C_{jms}$ is an overall normalisation constant
\bee
C_{jms}=(-1)^{j-s}  \left((j+s)!(j-s)!(j+m)!(j-m)!\right)^{\frac 12 },
\eee
and
\bee
j\in\frac 1 2 \mathbb{N}^0, \quad s, m=-j,-j+1,\ldots,j-1,j.
\eee
The summation index  $k$ runs over the values so that the factorials are well defined.
These functions are orthonormal and satisfy
\begin{equation}
\label{FI}
\Delta_{S^3}Y^j_{sm} = -j(j+1)Y^j_{sm}, \ \ iZ_3Y^j_{sm} = mY^j_{sm}, \ \ iX_3Y^j_{sm} = sY^j_{sm}.
\end{equation}
as well as
\begin{align}
iX_+ Y^j_{sm}= \sqrt{j(j+1)-s(s+1)}Y^j_{s+1,m}, \qquad
iX_- Y^j_{sm}= \sqrt{j(j+1)-s(s-1)}Y^j_{s-1,m}.
\end{align}

Using the explicit expression for the Dirac operator  given in  \eqref{S3Diracungauged}, one deduces that  
\begin{equation}
\Psi^{j+}_{sm} = \frac{1}{2j+1}\begin{pmatrix} \sqrt{j+s+1} Y^j_{sm} \\ \sqrt{j-s} Y^j_{s+1,m}\end{pmatrix}
\end{equation}
is an eigenspinor of $\Dslash_{S^3}$ with eigenvalue 
\bee
\lambda_+=  \frac 1 \ell\left(\frac 1 2 +\left(2j+1 \right)\right),
\eee
and that
\begin{equation}
\Psi^{j-}_{sm} =  \frac{1}{2j+1}\begin{pmatrix} -\sqrt{j-s} Y^j_{sm} \\ \sqrt{j+s+1} Y^j_{s+1,m}\end{pmatrix}
\end{equation}
is an eigenspinor of $\Dslash_{S^3}$ with eigenvalue 
\bee
\lambda_-=  \frac 1 \ell\left(\frac 1 2 -\left(2j+1 \right)\right).
\eee
The degeneracy of each eigenvalue is $(2j+1)(2j+2)$. 

We can now use a  trick introduced by Loss and Yau \cite{LY} to obtain zero-modes of the gauged Dirac operator from general eigenmodes of the ungauged Dirac operator. Setting 
\begin{equation}
\frac{2}{\ell} A_i = \lambda \frac{\Psi^\dagger \tau_i \Psi}{\Psi^\dagger \Psi}, \quad i=1,2,3,
\end{equation}
where $\Psi\neq 0$
and using $\Psi^\dagger \tau_i \Psi \tau_i =2 \Psi \Psi^\dagger  -{\Psi^\dagger \Psi}\mathbb{I}$, 
one then has
\begin{equation}
\frac{2}{\ell}\vec{A}\cdot\vec{\tau }\Psi=\lambda\Psi.
\end{equation}
With $A= A_i\sigma_i$,  this implies
\begin{equation}
\Dslash_{S^3}\Psi =\lambda\Psi \Leftrightarrow  \Dslash_{S^3,A}\Psi =0.
\end{equation}
In general, one needs to check the validity of this result  at the zeros of  $\Psi$. We will do this in our application of linear zero-modes later. Assuming the zeros are dealt with, we can then 
apply  Lemma \ref{conformal} to obtain   magnetic zero-modes on Euclidean 3-space of the form
\bee
\label{linearmode}
\Psi_H= G\Omega^{-1} H^* \Psi^{j\pm}_{sm}.
\eee

\section{Vortex equations and magnetic zero-modes}
\subsection{Vortex equations on $S^3$}
\label{vorts3}
We are now ready to introduce the 3-dimensional geometries which  will lead us to the smooth vortex zero-modes promised in the Introduction, and provide the link with vortices on the 2-sphere. First, we  define  vortex configurations on the 3-sphere.
\begin{dfn}
Let $n$ be a positive integer,  $A$   be a 1-form on $S^3$
and $\Phi: S^3\rightarrow \C$  be a complex-valued function.
We say that the pair $(\Phi,A)$ is a vortex configuration on  $S^3$ with vortex number $2n-2$ if the following conditions hold:
\begin{enumerate}
\item Normalisation :
\bee
\label{normalise}
A(X_3)=n-1,
\eee
\item Equivariance:
\bee
\label{equiva}
 \qquad  {\mathcal L}_{X_3}A =0,  \quad i{\mathcal L}_{X_3}\Phi =(n-1)\Phi,
\eee
\item 
Vortex equations:
\begin{align}
\label{vortexs3}
(d\Phi +i A\Phi)\wedge \sigma = 0, \qquad F_A  = \frac i 2 (|\Phi |^2-1)\bar\sigma \wedge \sigma,
\end{align}
where $F_A=dA$.
\end{enumerate}
\end{dfn}
The definition is such that vortex configurations are  mapped into vortex configurations by abelian gauge transformations of the form 
\bee
\label{abeliangauge}
(\Phi,A) \mapsto (e^{-i\alpha}\Phi,A+d\alpha), \qquad \alpha \in C^\infty(S^3), \quad X_3\alpha=0.
\eee

Our vortex configurations on $S^3$ can   be interpreted as an equivariant description of vortices on $S^2$ as follows. The normalisation condition \eqref{normalise} means that $iA$ may be viewed as the connection 1-form on the total space $S^3/\Z_{2n-2}$ (Lens space) of a  $U(1)$ bundle over $S^2$ of degree $2n-2$. Comparing with \eqref{equidef} and referring to \cite{JS} for details, the equivariance requirement \eqref{equiva} means that $\Phi$ is the equivariant form of a   section of the associated line bundle (the ($2n-2)$-th power of the hyperplane bundle). In fact, we will show in  Sect.~\ref{popovsect}, Lemma~\ref{pulledvortex},  that the vortices on $S^2$  which are equivariantly described by our vortex configurations are  Popov vortices.

We note that contracting the first vortex equation  with $(X_3,X_-)$ and using \eqref{pairing} gives $X_3\Phi +iA(X_3)\Phi=0$, which is satisfied by virtue of the normalisation and equivariance condition. Contracting it with $(X_+,X_-)$ gives 
\bee
\label{useforDirac}
X_+\Phi +iA(X_+)\Phi =0.
\eee
We will return to this equation later in this section and also  in  Sect.~\ref{popovsect}. However,  we  first show that the  vortex equations on $S^3$ can be  interpreted  in terms of a flat non-abelian gauge field.

The following theorem shows that any vortex configuration can be expressed in terms of  the pull-back of the  Maurer-Cartan form $h^{-1}dh$ on $SU(2)$ via a  bundle map $U:S^3 \rightarrow S^3$ of the Hopf fibration covering a rational map $S^2\rightarrow S^2$. Since the Maurer-Cartan form encodes the frame \eqref{ONframe} of the round 3-sphere, its pull-back encodes the pull-back of the round metric with $U$. In that sense, this result is a  3-dimensional  version of   Baptista's interpretation of vortices   as deformed 2-dimensional geometry.  
\begin{thm} 
\label{flatvor}
 A vortex configuration of degree $2n-2$ on $S^3$  determines a gauge potential 
 for a flat $su(2)$ connection on $S^3$ satisfying the normalisation condition 
\bee
\label{nonabnorm}
\A(X_3)=nt_3
\eee
via the following expression: 
\bee
\label{abnonab}
\A=(A+\sigma_3)t_3 +\frac 1 2( \Phi\sigma t_-+ \bar{\Phi}\bar{\sigma} t_+).
\eee
A  gauge potential for a  flat $su(2)$ connection on $S^3$ satisfying \eqref{nonabnorm} and of the form \eqref{abnonab} can be trivialised as  $\A=U^{-1}dU$,  where  $U:S^3\rightarrow S^3 $ has degree $n^2$ and is a bundle map of the Hopf fibration, covering a rational map $R:S^2\rightarrow S^2$ of degree $n$.
 Up to a $U(1)$ gauge transformation \eqref{abeliangauge}, one can choose the bundle map $U$ to have  the form
\bee
\label{trivialisation}
U:(z_{1},z_{2})\mapsto \frac{1}{\sqrt{|P_{1}|^{2}+|P_{2}|^{2}}} \begin{pmatrix} P_{1} & -\bar{P}_2 \\  P_{2} & \phantom{-}\bar{P}_1\end{pmatrix},
\eee 
where $P_1,P_2$ are homogeneous polynomials of degree $n$ with no common zeros
\bee
P_1= a_0z_1^n + a_1z_1^{n-1}z_2+\ldots + a_n z_2^n, \qquad 
P_2= b_0z_1^n + b_1z_1^{n-1}z_2+\ldots + b_n z_2^n,
\eee
and $a_0,b_0,a_n,b_n$ all non-zero. 

The vortex configuration  $(\Phi,A)$ can be computed from the bundle map $U$ via
\bee
\label{UphiA}
U^*\sigma =\Phi\sigma, \qquad  A= U^*\sigma_3-\sigma_3,
\eee
and  is  given in terms of $P_1,P_2$  by
\bee
\Phi = \frac{P_1\partial_2 P_2 - P_2 \partial_2 P_1}   {z_1(|P_1|^2  + |P_2|^2)},
\eee
and 
\bee
\label{aformula}
A= (n-1)\sigma_3 + \frac{i}{2} X_-\ln(|P_1|^2  + |P_2|^2)\sigma - \frac{i}{2}  X_+\ln (|P_1|^2  + |P_2|^2) \bar{\sigma}.
\eee
\end{thm}
Our condition on  $a_0,b_0,a_n,b_n$    will turn out to be convenient in the discussion of Popov vortices in Sect.~\ref{singsect} and facilitates comparison with the treatment in \cite{Manton1}. 

\begin{proof} Suppose $(\Phi,A)$ is a vortex configuration of degree $2n-2$. It is easy to check that, for a gauge potential of the form  \eqref{abnonab}, the normalisation \eqref{normalise} implies \eqref{nonabnorm}. 
 The flatness condition $d\A + \A\wedge \A=0$ for a gauge potential of the form \eqref{abnonab} is equivalent to
\bee
d(\Phi \sigma) + i(A+\sigma_3) \wedge \Phi \sigma =0,   \qquad dA = \frac{i}{2}(|\Phi|^2 -1)\bar{\sigma}\wedge \sigma, 
\eee 
which, using \eqref{usefulsigma},  is equivalent to the vortex equations  \eqref{vortexs3}.
The  equivariance condition \eqref{equiva} for vortex configurations is equivalent to 
\bee
{\mathcal L}_{X_3} \A = n[\A,t_3],
\eee
but this holds automatically for a flat gauge potential satisfying the  normalisation \eqref{nonabnorm} since, for a flat gauge field, 
\bee
{\mathcal L}_{X_3} \A =  D_\A \A(X_3).
\eee

A flat  and smooth $SU(2)$ gauge potential $\A$ on $S^3$    can always be globally trivialised in terms of a function $U:S^3 \rightarrow SU(2)$ as   $\A=U^{-1}dU$.  We now show  that  the vortex form \eqref{abnonab} and the normalisation \eqref{nonabnorm} forces the trivialising map to be  a bundle map covering a rational map of degree $n$. 
The  normalisation \eqref{nonabnorm} requires
\bee
\label{Udegree}
X_3 U =nUt_3, \quad \text{or} \quad U(he^{\gamma t_3}) = U(h)e^{n\gamma t_3}, \quad \gamma \in [0,4\pi).
\eee
This equivariance condition has important topological consequences. It implies that the map $\pi \circ U$ is constant along fibres of the Hopf fibration and    determines a map $S^2 \rightarrow S^2$; in the parametrisation of $U$ in terms of two  functions $P_1,P_2$ which do not vanish simultaneously as in \eqref{trivialisation} (but without assuming that $P_1,P_2$ are polynomials) this map is simply the quotient $P_2/P_1$. In terms of our stereographic coordinate $z$ for $S^2$ and the section $s$ in \eqref{locsec}  we define  
\bee
R = \pi \circ U \circ s,
\eee
and have the following   commutative  diagram (where we have not carefully distinguished between $S^2$ and our coordinate chart $\C$  for it):
\bee
\label{URcom}
\begin{CD}
S^3  @>U>> S^3 \\
@VV \pi V   @VV \pi V \\
 S^2 @> R>> S^2.
\end{CD}
\eee
By virtue of \eqref{Udegree}, the map $\pi \circ U:S^3\rightarrow S^2$ has Hopf number $n^2$: the pre-image of any point on $S^2$  is an  $n$-fold cover of a circle which links with each of the $n$ circles in the pre-image of another point exactly once. It follows that   the map $U$ has degree $n^2$ and the map $R$ covered by $U$ has degree $n$.

Continuing in a  parametrisation of $U$ in terms of two  functions $P_1,P_2$ but still not  assuming that $P_1,P_2$ are polynomials, the condition \eqref{Udegree} implies 
\bee
\label{polydeg}
2i X_3P_1 = nP_1, \quad  2i X_3P_2 = nP_2.
\eee
Since 
\bee
U^{-1} dU = U^*\sigma_3 t_3 +\frac 1 2( U^*\sigma t_-+ U^*\bar{\sigma} t_+),
\eee
we obtain a potential in the vortex gauge \eqref{abnonab} if and only if
\bee
 U^*\sigma  = \Phi \sigma
\eee
for some function $\Phi:S^3 \rightarrow \C$. Using \eqref{pairing}, we therefore need to show that 
\bee
(U^{*}\sigma)(X_{3})=0, \qquad( U^{*}\sigma)(X_{+}) =0.
\eee
The first of these  follows from   \eqref{polydeg}, since 
\bee
\label{X3cond}
(U^{*}\sigma)(X_{3})= \frac{2i}{|P_1|^2+ |P_2|^2}(P_1X_3P_2-P_2X_3P_1).
\eee
To analyse the  second condition, note that  
\bee
(U^{*}\sigma)(X_{+})=\frac{2i}{|P_1|^2+ |P_2|^2}(P_{1}X_{+}P_{2}-P_{2}X_{+}P_{1})=\frac{2i}{|P_1|^2+ |P_2|^2}P_{1}^{2}X_{+}\left(\frac{P_{2}}{P_{1}}\right),
\eee
with the last equality holding where $P_1\neq 0$.  As noted above,  the ratio $P_2/P_1$   defines a function (section of the trivial bundle $H^0$)  on $S^2$. According to the commutative diagram \eqref{equicom}, $X_+(P_2/P_1)=0$ means that the pull-back   $R=s^*(P_2/P_1)$  is, in fact,  a  holomorphic function  where it is defined. Thus $R$ has to be a holomorphic map $S^2\rightarrow S^2$ of degree $n$, which means it must be a rational map, as claimed. 

These conditions are clearly satisfied  when $P_1$ and $P_2$ are  the homogeneous polynomials given in \eqref{trivialisation}. In that case, the rational map is explicitly given by 
\bee
\label{Rfirst}
R(z)=\frac{p_2(z)}{p_1(z)}, 
\eee
where 
\bee
p_1(z)= a_0+a_1z+\ldots +a_nz^n, \quad p_2(z)= b_0+b_1z+\ldots +b_nz^n.
\eee
In order for \eqref{Rfirst} to be a  map of degree $n$ we require at least one of $a_n,b_n$ to be non-zero (so that the maximum of the degrees of $p_1$ and $p_2$ is $n$)  and at least one of $a_0, b_0$ to be non-zero (so that we cannot reduce the degree by cancellation). We can then arrange for all of  $a_0, b_0,a_n,b_n$ to be non-zero by left-multiplying  $U$ with a constant $SU(2)$ matrix if necessary; this does not affect $\A$ and therefore leaves the vortex configuration unchanged.

 Fixing  $U$ to be the trivialisation in terms of the   polynomials  $P_1,P_2$ in \eqref{trivialisation},   we can define a new trivialisation
\bee
\tilde U = U e^{\alpha t_3}, \quad \alpha \in C^\infty(S^3),\qquad  X_3 \alpha =0.
\eee
This also satisfies \eqref{Udegree}, and leads to the same rational map $R$. The  non-abelian gauge potential $\tilde \A= \tilde U^{-1} d\tilde U$   differs  from $\A = U^{-1} dU$ by the  gauge transformation \eqref{abeliangauge}, as claimed.

Continuing with     $P_1$ and $P_2$ being  homogeneous polynomials in $z_1,z_2$ of degree $n$,
we obtain the claimed formula for the vortex field $\Phi$  from
\bee
\Phi =\frac{1}{2} (U^{*}\sigma)(X_{-}) = \frac{P_1\partial_2 P_2 - P_2 \partial_2 P_1}   {z_1(|P_1|^2  + |P_2|^2)},
\eee
noting that 
\bee
\frac{P_1\partial_2 P_2 - P_2 \partial_2 P_1}   {z_1}
\eee
is a homogeneous polynomial in $z_1,z_2$ of degree $2n-2$ and non-singular: the term of order $z_2^{2n-1}$, which could potentially cause a singularity when divided by $z_1$, vanishes.

The derivation of the expression for $A$ is a straightforward calculation, which makes  use of 
\bee
(z_1\partial_1 + z_2\partial_2) (|P_1|^2+|P_2|^2) = n(|P_1|^2+|P_2|^2). 
\eee 
One finds
\bee
U^*\sigma_3(X_3) =n, \;\; U^*\sigma_3(X_+)=-iX_{+}\ln (|P_1|^2 +|P_2|^2)  , \;\;  U^*\sigma_3(X_-)= iX_{-}\ln  (|P_1|^2+|P_2|^2), 
\eee
which, with \eqref{pairing},  implies  \eqref{aformula}.
\end{proof}

In order to make contact with discussions in the literature related to  the potential $A$  we note an  expression  for $A$ in terms of polar coordinates, for later use.
\begin{lem}
Suppose $(\Phi,A)$ is a vortex configuration of vortex number  $2n-2$ on $S^3$ and consider  
the  modulus-argument parametrisation  
 \bee
 \label{Phimodarg}
\Phi = e^{\frac{M}{2}+i\chi},
\eee
valid away from the  (generically $2n-2$) zeros of $\Phi$. Then the gauge potential $A$ in  \eqref{UphiA}  can be
expressed via the formula
\bee
\label{Amodarg}
A= -\frac{\ell}{4}  \star (\sigma_3 \wedge  dM) -d\chi,
\eee
valid away from the zeros of $\Phi$.
\end{lem}

\begin{proof} 
Observe that, away from the zeros of $\Phi$, we can write  \eqref{aformula} as 
\bee
\label{aformula_phi}
A= (n-1)\sigma_3 - \frac{i}{2}X_-\ln \bar{\Phi}\sigma + \frac{i}{2}X_+\ln \Phi\bar{\sigma}.
\eee
Inserting the parametrisation \eqref{Phimodarg} leads to
\bee
\label{firstform}
 A -(n-1)\sigma_3=  \left(- \frac {i}{4} X_-M  - \frac 12 X_-\chi \right) \sigma +
\left(\frac {i}{4} X_+ M  -\frac 12 X_+\chi \right)\bar{\sigma}.
\eee
With the Hodge-$\star$ relative to the orientation  \eqref{orient}, we have 
\bee
\star (\sigma_3 \wedge \sigma) = i\frac{2}{\ell} \sigma, \quad \star (\sigma_3 \wedge\bar  \sigma) = -i\frac{2}{\ell} \bar \sigma,
\eee
so that 
\bee
\label{Mresult}
- \frac {i}{4} X_-M \sigma + \frac {i}{4} X_+ M\bar{\sigma}=-\frac{\ell}{4}  \star (\sigma_3 \wedge  dM),
\eee
where we have used that for any differentiable $f:S^3\rightarrow \C$,
\bee
df =\frac 12  X_-f \sigma + \frac 12 X_+f \bar{\sigma}  +  X_{3}f\sigma_{3}, \label{derivdecomp}
\eee
Turning to the terms involving $\chi$, using \eqref{derivdecomp} and deducing from \eqref{equiva} that $X_3\chi = 1-n$,  we conclude that
\bee
\label{chiresult}
d\chi =\frac 12  X_-\chi \sigma + \frac 12 X_+ \chi \bar{\sigma}  -  (n-1)\sigma_3.
\eee
Combining \eqref{firstform} with  \eqref{Mresult} and \eqref{chiresult}  we arrive at the claimed expression for the gauge potential \eqref{UphiA} in terms of the modulus and argument of the field $\Phi$. 
\end{proof}

\subsection{Magnetic zero-modes from vortices}
\label{magfields}
 We are now ready to explain how one can construct magnetic zero-modes of the Dirac operator on the 3-sphere and on Euclidean 3-space  from vortex configurations on the 3-sphere. 
 We define spinorial  vortex zero-modes  as follows.
\begin{dfn}
\label{spinorvortex} A pair $(\Psi,A)$ of a spinor $\Psi$ and a 1-form $A$ on $S^3$ is said to be a  vortex zero-mode of the Dirac equation on $S^3$ if 
\begin{align}
\label{spinvoreq}
\Dslash_{S^3,A}\Psi =0 , \qquad F_A= \frac{4i}{\ell}\star \Psi^\dagger h^{-1}dh\Psi + \frac 1 4 \sigma_1\wedge \sigma_2,
\end{align}
where $\star$ is the Hodge star operator on $S^3$ with respect to the metric \eqref{S3met} and orientation \eqref{orient}. 
\end{dfn}

\begin{thm}
\label{dirvor}
Suppose $(\Phi,A)$ is a vortex configuration on $S^3$. Then the pair 
\bee
\Psi = \begin{pmatrix} \Phi \\ 0  \end{pmatrix}, \qquad A'= A+\frac 3 4 \sigma_3,
\eee
is a  vortex  zero-mode   $(\Psi,A')$  on $S^3$.
\end{thm}

\begin{proof}
The spinor given in the theorem is a zero-mode of the gauged Dirac equation if 
\bee
\left(iX_3-A'_3 +\frac 3 4 \right)\Phi = 0 \quad \text{and} \quad X_+\Phi + iA'_+\Phi =0.
\eee
However, $A'_3= A'(X_3)= (n-1)+ \frac 3 4$ so that the first of these equations follows from \eqref{equiva}. The second follows from $A'(X_+) = A(X_+)$ and \eqref{useforDirac}. 
Turning to the non-linear equation, we note that, 
for a spinor of the form given in the theorem, 
\bee
 \frac{4i}{\ell}\star \Psi^\dagger h^{-1}dh\Psi =  \frac{4i}{\ell}|\Phi|^2 \star \left(-\frac{i}{2}\sigma_3\right)
=  |\Phi|^2\sigma_2\wedge \sigma_1.
\eee 
Moreover,
\bee
F_{A'} = F_A + \frac 3 4 \sigma_2\wedge \sigma_1 = \left(|\Phi|^2 - \frac 1 4 \right)\sigma_2\wedge \sigma_1,
\eee
so that the non-linear equation in the definition of a vortex zero-mode follows.
\end{proof}

We can pull back the vortex  zero-modes of the Dirac equation on $S^3$ to $\R^3$ using Lemma \eqref{conformal}, but  we also need to understand how the non-linear equation behaves under this pull-back. It turns out that the resulting equations take their simplest form in vector notation for gauge potentials and their magnetic fields, i.e., when expanding a 1-form on $\R^3$ as $A=\vec{A}\cdot d\vec{x}$ and defining the magnetic field  vector field via $dA=\frac 12 \epsilon_{jkl} B_j dx_k \wedge dx_l$ or $\vec{B}=\nabla \times \vec{A}$. 

The  magnetic field  corresponding to  the inhomogeneous term is given by
\bee
\label{backfield}
\frac{1}{4}H^*(\sigma_1\wedge \sigma_2)= \frac{4\ell^2}{(\ell^2+r^2)^2}
 \star_{\R^3}\theta_3= 
\frac 12 \epsilon_{jkl} b_{j} dx_k\wedge dx_l, \quad 
\vec{b}=\frac{4\ell^{2}}{(\ell^{2}+r^{2})^{3}}\begin{pmatrix}
2(x_1x_3-\ell x_2)\\ 2(x_2x_3+\ell x_1)\\ \ell^{2}-r^{2}+2x_3^{2}
\end{pmatrix},
\eee
where we used \eqref{theta3exp}. The 
integral lines  of $\vec{b}$  are the fibres of the Hopf fibration \eqref{Hopf}; they are  plotted in  Fig.~\ref{background}.

\begin{figure}[!htbp]
  \centering
\includegraphics[width=7truecm]{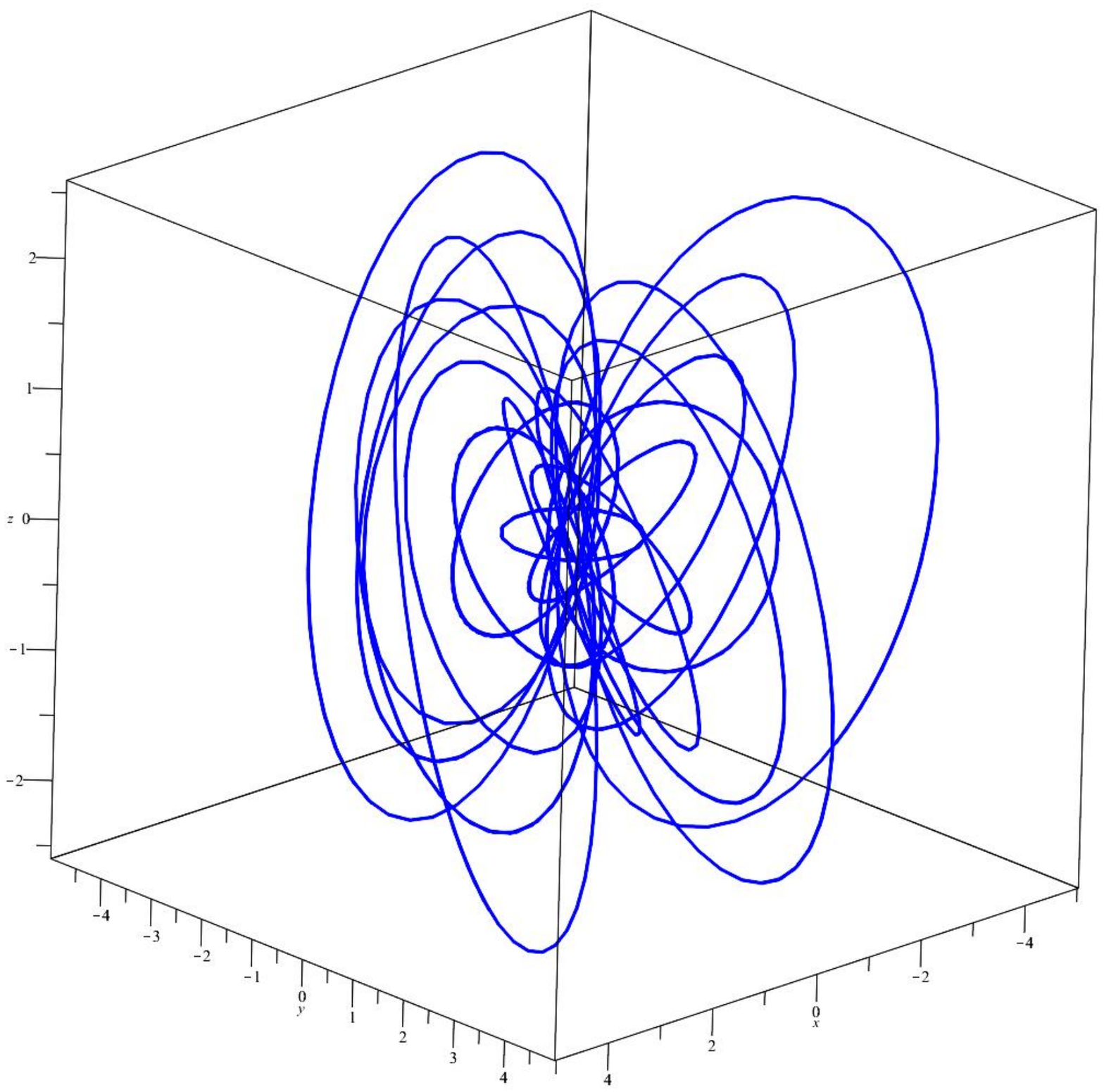}
 \caption{A plot of some of the integral curves of the background field $\vec{b}$  given in \eqref{backfield}}
 \label{background}
\end{figure}

We claim that vortex zero-modes of the Dirac equation pull back to solutions of the following coupled equations in $\R^3$:
\bee
\label{coupledr3}
\Dslash_{\R^3,A}\Psi =0,  \qquad \vec{B} = -\Psi^\dagger \vec{\tau}\Psi + \vec{b}.
\eee
We state this result as follows.
\begin{cor}
Any pair of homogeneous polynomials $P_1,P_2:\C^2 \rightarrow \C$ of the same degree and  without common zeros uniquely determines a smooth and square-integrable  magnetic zero-mode    of the Dirac operator in Euclidean 3-space which satisfies the coupled equations \eqref{coupledr3}. 
\end{cor}
\begin{proof} Combining $P_1$ and $P_2$ into an $SU(2)$ matrix yields a vortex configuration $(\Phi,A)$ on $S^3$ according to the prescription of Theorem \ref{flatvor}. Such a vortex configuration  defines  a vortex zero-mode $\Psi$ of the Dirac operator on $S^3$  coupled to $A'= A+ \frac 34  \sigma_3$  according to  Theorem \ref{dirvor}. Implementing  the conformal change  to $\R^3$ according to Lemma \ref{conformal} produces the  magnetic zero-mode  
\bee
\label{PPsi}
\Psi_H=  \Omega^{-1} G H^*\Psi = \Omega^{-1} G\begin{pmatrix}H^*\Phi \\ 0 
 \end{pmatrix}
 \eee  
 on $\R^3$  of the Dirac operator coupled to $H^*A'=\vec{A'}_H\cdot d\vec{x}$.

We also need to understand the pull-back of the non-linear equation in \eqref{spinvoreq}.
The quadratic term in the zero-mode $\Psi$ on $S^3$ is  
\bee
 \frac{4i}{\ell}\star \Psi^\dagger h^{-1}dh\Psi=  \frac 12 \epsilon_{jkl} \Psi^\dagger\tau_j \Psi \sigma_l\wedge\sigma_k,
\eee
and using  \eqref{pullbacksigma} and \eqref{rotatingsigma} we deduce that pulling back   with $H$ yields
\bee
H^*\left( \frac{4i}{\ell}\star \Psi^\dagger h^{-1}dh\Psi\right)=\frac{1}{2\Omega^2} \epsilon_{jkl} H^*\Psi^\dagger G^{-1}\tau_j G H^*\Psi dx_l\wedge dx_k = -\frac 12 \Psi^\dagger_H\tau_j\Psi_H\epsilon_{jkl} dx_k\wedge dx_l, 
\eee
where $\Psi_H$ is related to $\Psi$ as  defined in \eqref{r3spinor}.
Finally expanding the pull-back of the field strength in the same coordinates 
\bee
H^*F_{A'}=\frac 12 \epsilon_{jkl} (B'_H)_j dx_k \wedge dx_l, \quad \text{so that} \quad \vec{B'}_H =\nabla \times \vec{A'}_H,
\eee
the non-linear equation  in the definition \eqref{spinorvortex} pulls back to 
\bee
\label{Bb}
\vec{B'}_H = - \Psi_H^\dagger \vec{\tau} \Psi_H + \vec{b},
\eee
as claimed.

The spinor $\Psi_H$  in  \eqref{PPsi} and the gauge potential $H^*A'$ are manifestly smooth, being the pull-back with smooth maps of smooth functions on $S^3$. As the pull-back of a smooth function on $S^3$, $H^*\Phi$ is bounded and has a (finite) limit as $r\rightarrow \infty$.
It follows that 
\bee
|\Psi_H^\dagger \Psi_H| \leq \frac{C}{\Omega^2},
\eee
 for some positive constant $C$, which ensures that $\Psi_H$  is square-integrable with respect to the Euclidean  measure \eqref{euclid}.  Since $|\Psi^\dagger \vec{\tau} \Psi|= |\Psi^\dagger \Psi|$ for any spinor $\Psi$, it follows that the vector field  $ \Psi_H^\dagger \vec{\tau} \Psi_H$ is also square-integrable. The square-integrability of $ \vec{B'}_H $ then follows from the square-integrability of  $ \vec{b}$  and the relation \eqref{Bb}.
\end{proof}

The  coupled equations \eqref{coupledr3}  have appeared in the literature in various contexts and  deserve  a few comments.
There are various ways of stating these equations. Re-scaling the spinor by a factor of $\sqrt 2$ leaves the linear equations unchanged, but changes the quadratic term in the non-linear equation into the spin density 
\bee
\label{density}
\vec{\Sigma}  =\frac 12 \Psi^\dagger \vec{\tau} \Psi.
\eee
Changing to the charge-conjugate spinor
\bee
\label{cc}
\Psi^c = i\tau_2\Psi,
\eee
turns our equations into the equivalent set of equations
\bee
\label{coupledr3cc}
\Dslash_{\R^3,-A} \Psi^c =0, \qquad
\vec{B}  =  (\Psi^c)^\dagger\vec{\tau}\Psi^c +\vec{b}.
\eee 
The  equations \eqref{coupledr3} have been discussed in the literature as the dimensionally reduced Freund equations \cite{AMN1}, while their charge-conjugates have appeared as the variational equations of  a particular Dirac-Chern-Simons action \cite{AMN3}.

The magnetic field on $\R^3$ obtained from the pair of complex polynomials $P_1,P_2$ can be written in terms of the  Hopf map $\pi$ and the maps $H$ \eqref{hdef} and $U$ \eqref{trivialisation}  as 
\bee
\label{Ranadafield}
F=(\pi \circ U \circ H)^*\mathcal{R},
\eee
where $\mathcal{R}$ is the area form on the 2-sphere of unit radius \eqref{2darea}, which will play an important role in the next section. This is an example of the magnetic field introduced  by Ra\~{n}ada  in \cite{Ranada1} and discussed more recently  in \cite{IB}.  It has interesting topological properties inherited from those of the map $U:S^3\rightarrow S^3$, which, as explained after diagram \ref{URcom},  has topological degree $n^2$. As also explained there,  $\pi \circ U:S^3\rightarrow S^2$ has Hopf number $n^2$.  As discussed in \cite{Ranada2}, this implies that the magnetic field  \eqref{Ranadafield} has linking number one and  (magnetic) helicity $n^2$.

In order to compare with solutions for \eqref{coupledr3} previously obtained in the literature, we also pull back the modulus-argument expression \eqref{Amodarg} to $\R^3$, to find 
\bee
H^{*}A=-\frac{\ell}{4}H^{*}\left(\star(\sigma_{3}\wedge dM)\right)-d(H^*\chi).
\eee
The operation 
\bee
\star\sigma_3\wedge\, : \Lambda^1(S^3)\rightarrow \Lambda^1(S^3)
\eee
is linear; it annihilates $\sigma_3$ and  acts as a complex structure on the cotangent space orthogonal to $\sigma_3$   by mapping 
\bee
\star \sigma_3\wedge \, :  \sigma  \mapsto (2 i/  \ell) \sigma.
\eee
It pulls back  to the map $(\theta_1+i\theta_2) \mapsto  (2 i/  \ell)( \theta_1+i\theta_2)$, which we can write as 
\bee
-\frac 2 \ell \star_{\R^3}\theta_3\wedge\,: \Lambda^1(\R^3)\rightarrow \Lambda^1(\R^3).
\eee
Therefore 
 \bee
 \label{modarggauge}
H^*A=-\frac{1}{2}\star_{\R^{3}}\left(d (H^*M) \wedge\theta_{3}\right)-d (H^*\chi).
\eee

In \cite{LY}, Loss and Yau  showed that, for spinors on $\R^3$  whose spin density \eqref{density} has vanishing divergence, one can always find a gauge field $A$ so that the given spinor is a zero-mode of the Dirac operator on $\R^3$ coupled to $A$. They gave an explicit formula, valid where the spinor does not vanish:
\bee
\label{LossYauformula}
A_\Psi= -\frac 1 2 \star    \frac{d (\Psi^\dagger d  \vec{x} \cdot \vec{\tau} \Psi)}  {\Psi^\dagger \Psi} - \frac{ \text{Im}(\Psi^\dagger d \Psi)}  {\Psi^\dagger \Psi}.
\eee
The relation to our expression \eqref{modarggauge} is as follows.
\begin{lem}
Let  $(\Phi,A)$ be a vortex configuration on $S^3$, and $\Psi_H$ the corresponding zero-mode of the Dirac operator on $\R^3$  given in \eqref{PPsi}. Then the spin density of $\Psi_H$ is divergenceless, and the corresponding Loss-Yau potential is given by 
\bee
A_{\Psi_H} = H^*A+ \frac 3 4 H^*\sigma_3.
\eee
\end{lem}
\begin{proof}
As shown above, the spinor $\Psi_H$ \eqref{PPsi}  and the gauge potential $H^*A+ \frac 3 4 H^*\sigma_3$ satisfy the coupled equations \eqref{coupledr3}.  By  virtue of the non-linear equation \eqref{Bb}, the spin density for  $\Psi_H$  automatically has  vanishing divergence. 
Using the expression \eqref{PPsi} and the modulus-argument decomposition \eqref{Phimodarg} pulled back to $\R^3$,i.e., 
\bee
H^*\Phi = e^{\frac{1}{2}H^*M + i H^*\chi},
\eee
one then  computes
\bee
A_{\Psi_H}= -\frac{1}{2}\star \left(d (H^*M) \wedge\theta_{3}\right)-d (H^*\chi) -(1,0)
\left(\frac {\Omega^2}{2} \star d \left( \frac{G^{-1}d  \vec{x} \cdot \vec{\tau}G }{\Omega^2}\right) -i G^{-1} dG\right)\begin{pmatrix}
1\\0
\end{pmatrix}
\eee
where all Hodge star operations now refer to $\R^3$. 
The first two terms combine to the expression \eqref{modarggauge} for  $H^*A$.
The first term  inside the expectation value can be re-written, using \eqref{rotatingsigma} and \eqref{pullbacksigma}:
\begin{align}
\frac {\Omega^2}{2} \star d \left( \frac{G^{-1}d  \vec{x} \cdot \vec{\tau}G }{\Omega^2}\right)
&=- \Omega^ 2 \star d \left( \frac{i}{\Omega} H^{-1} dH \right)  \nonumber \\
&= -i H^{-1}dH  +i\star ( d\Omega\wedge H^{-1} dH), 
\end{align}
where we used
\bee
\star d(H^{-1} dH)=\frac 1 \Omega H^{-1} dH.
\eee
Next we use the relation \eqref{GH} to  express $G^{-1}dG $ in terms of $H$ and $\Omega$, to deduce
\bee
\frac {\Omega^2}{2} \star d \left( \frac{G^{-1}d  \vec{x} \cdot \vec{\tau}G }{\Omega^2}\right) -i G^{-1} dG =
-\frac {3i} {2} H^{-1}dH.
\eee
Then the observation
\bee
(1,0)\left( \frac {3i}{ 2} H^{-1}dH \right)\begin{pmatrix}
1\\0
\end{pmatrix} = \frac{3}{4}H^*\sigma_3,
\eee
 completes the proof.
\end{proof}

The formula \eqref{LossYauformula} is the starting point of several treatments   in the literature of magnetic zero-modes, particularly in the papers \cite{AMN1,AMN2,AMN3} by Adam, Muratori and Nash (AMN). The AMN construction gives magnetic zero-modes in terms of solutions of the Liouville equation. However, by effectively  pulling back local expressions for sections on $S^2$ to $S^3$ via the Hopf map it introduces additional singularities which we will discuss  in more detail  in Sect. \ref{singsect}.

\subsection{Zero-mode combinatorics}
It is natural to wonder if the linear magnetic zero-modes \eqref{linearmode} and the vortex zero-modes \eqref{PPsi} can be combined  to produce new zero-modes. This is indeed possible when one picks $s=j$ in \eqref{linearmode}, and notes that, according to \eqref{Y}, 
\bee
Y^j_{jm} = C_{jmj}z_1^{j-m} z_2^{j+m},  \quad j\in \frac 12 \N^0,
\eee
so that a linear combination of such functions gives another homogeneous polynomial 
\bee
P= A_0z_1^{2j}+ A_1z_1^{2j-1}z_2 + \ldots + A_{2j}z_2^{2j}, 
\eee
of degree $2j$. Since such a polynomial satisfies $iX_3P= j P$ and $X_+P=0$, it is easy to check that one can combine it with a vortex configuration $(\Phi,A)$ of degree $2n-2\geq 0$ to get a solution 
\bee
\label{holmode}
\Psi = G\Omega^{-1} H^* \begin{pmatrix} P\Phi \\ 0 \end{pmatrix}, \qquad H^*A'= H^*A + \left(j+ \frac{3}{4}\right)H^*\sigma_3,
\eee 
of the coupled Dirac equations \eqref{coupledr3} on $\R^3$. Physically, the inclusion of $P$ in the spinor adds a multiple of the background field $\vec{b}$ to the solution. 

There is  an obvious mirror version  of all our solutions in the anti-holomorphic world: for negative $n$ and $s=-j$, one can write down vortex configurations $(\bar \Phi,-A)$ in terms of anti-holomorphic polynomials, and obtain corresponding Dirac zero-modes 
\bee
\Psi = G\Omega^{-1} H^* \begin{pmatrix}  0 \\ \bar{P} \bar{\Phi} \end{pmatrix},
\eee 
of the Dirac operator coupled to $ -H^*A - \left(j+ \frac{3}{4}\right)H^*\sigma_3.$
These are nothing but the charge-conjugates \eqref{cc}   of the holomorphic solutions \eqref{holmode}.

\section{Popov vortices on $S^2$ and Cartan connections}
\subsection{Popov vortices from vortices on $S^3$}
\label{popovsect}
We now turn to the promised explanation of the link between our  vortex equations on $S^3$ and vortex equations on $S^2$ whose solutions are called Popov vortices.  Before we write down the equations,  we introduce our notation for the round geometry of the 2-sphere. 

In  a stereographic coordinate $z$ defined by projection from the south pole, the round metric of a 2-sphere  of arbitrary radius $\rs$  is 
\bee
ds^2 = \frac{4\rs^2 dz d\bar{z}} {(1+|z|^2)^2},
\eee
and  a  possible complexified frame  field  $e=e_1+ie_2$ is given by
\bee
\label{compframe}
e=\frac{2\rs}{1+|z|^{2}}dz,\quad \bar{e}=\frac{2\rs}{1+|z|^{2}}d\bar{z}.
\eee
In terms of this frame field, the structure equations can be written as a single complex equation 
\bee
\label{structure}
 de-i\Gamma\wedge e =0,
\eee
which determine the  spin connection 1-form $\Gamma$ as 
\bee
\label{spincon}
\Gamma = i\frac{zd \bar{z} -\bar{z} dz}{1+|z|^2}.
\eee
The topology of the 2-sphere does not permit a globally defined frame, and 
one checks that the  frame \eqref{compframe} is singular at  $z=\infty$ (the south pole) by switching to $\zeta =1/z$ \cite{JS} and noting that $e$ behaves likes $\bar{\zeta}^2/|\zeta|^2 d\zeta$  near $\zeta =0$;  $\Gamma$, too, has a singularity at $z=\infty$.
In our chart, the Riemann curvature form is
\bee
\mathcal R = d\Gamma, 
\eee
and is related to the frame via the usual Gauss equation
\bee
\label{Gauss}
\mathcal R =  K e_1\wedge e_2 = \frac{i}{2\lambda^2}e\wedge \bar{e} ,
\eee
where $K=1/\rs^2$ is the Gauss curvature. Thus 
\bee
\label{2darea}
\mathcal R = 2i\frac{dz \wedge d\bar{z}  }{(1+|z|^2)^2},
\eee
which integrates to $4\pi$.

In \cite{Manton1}, $\lambda =\sqrt{2}$ is chosen and  the Popov equations are expressed in terms of the associated K\"ahler form. However, as we shall see it is more natural to write them in terms of the Riemann curvature form. A Popov vortex is   defined on a principal $U(1)$ bundle of degree $2n-2$ over the 2-sphere. It is a pair $(\phi,a)$ of a connection $a$ on this bundle and a  section $\phi$  of  the associated complex line bundle.  
With $a= a_z dz + a_{\bar{z} }d\bar{z} $ and 
$f=f_{z\bar z} dz\wedge  d\bar{z}=da$, the vortex equations in \cite{Manton1} are 
\bee
\label{Popov}
\partial_{\bar z} \phi -i a_{\bar z} \phi =0, \qquad f= (|\phi|^2 -1) \mathcal R.
\eee
As also explained in \cite{Manton1},  solutions are obtained  from   rational maps $R: S^2 \rightarrow S^2$ of degree $n$ which, in our coordinate $z$, take the form \eqref{Rfirst}.
 The Popov vortices are determined  by
\bee
\label{popoveq}
a= R^*\Gamma-\Gamma, \qquad R^*e= \phi e. 
\eee
The second of these equations determines $\phi$ as 
\bee
\phi = \frac{R' (1+|z|^2)}{1+|R|^2} = (p_2' p_1- p_1' p_2)\frac{1+|z|^2}{|p_1|^2  + |p_2|^2}\frac{\bar{p}_1}{p_1},
\eee
which has singularities at the zeros of $p_1$ which we will discuss below (see also \cite{Manton1}).
Note also that \eqref{popoveq} implies that 
\bee
R^*(e\wedge \bar e)= |\phi|^2 e\wedge \bar e, 
\eee
so that 
\bee
f= d(R^*\Gamma -\Gamma) = R^*\mathcal R - \mathcal R = (|\phi|^2 -1) \mathcal R.
\eee
follows immediately. 

We would like to relate the Popov equations and their solutions to vortices on the 3-sphere studied in Sect.~\ref{vorts3}. As reviewed in Sect.~\ref{conventions},
the Hopf projection  $S^3\simeq SU(2)\rightarrow S^2$ in terms of  complex  coordinates  $(z_1,z_2)$  for  $h\in SU(2)$ (see \eqref{complexp})  and the complex  stereographic coordinate $z$ on $S^2$ is $\pi:  h \mapsto z_2/z_1$, and  
a local section of this bundle is given by \eqref{locsec}.
\begin{lem}
\label{pulledvortex}
The pull-back of the vortex equations  \eqref{vortexs3} on $S^3$ via the section $s$ \eqref{locsec} yields the  Popov equations  \eqref{Popov} up to a singular gauge transformation. 
\end{lem}
\begin{proof}  With $U:   SU(2) \rightarrow SU(2)$ defined in terms of polynomials $P_1,P_2$ as in \eqref{trivialisation}  we define $R:S^2 \rightarrow S^2$  in our stereographic chart by 
\bee
\label{Rdef}
R=   \pi \circ U \circ s.
\eee
This is the rational map already encountered in the proof of Theorem~\ref{flatvor}, see also the diagram \eqref{URcom}; it is  of the form \eqref{Rfirst}. 
One checks that  
\bee
\label{basic}
\pi^*\Gamma= -\sigma_3+ i d\ln\frac{z_1}{\bar{z}_1},
\eee
so that  that the pull-back with \eqref{locsec} gives
\bee
s^*\sigma_3 =-\Gamma.
\eee
Pulling  back \eqref{basic}  further  with $U$
\bee
U^*\pi^*\Gamma= -U^*\sigma_3+ i d\ln\frac{P_1}{\bar{P}_1},
\eee
and then  with $s$,
\bee
s^*U^*\pi^*\Gamma= -s^*U^*\sigma_3+ i d\ln\frac{p_1}{\bar{p}_1}
\eee
 we deduce from the definition of $R$ that  
\bee
R^*\Gamma= -s^*U^*\sigma_3+ i d\ln\frac{p_1}{\bar{p}_1}. 
\eee
Thus we find that the pull-back of the 1-form $A=U^*\sigma_3-\sigma_3$ is 
\bee
s^*A = s^*U^*\sigma_3 - s^*\sigma_3 = -R^*\Gamma +   \Gamma+ i d\ln\frac{p_1}{\bar{p}_1}.
\eee
Also noting
\bee
\label{Phipull}
s^*\Phi = \frac{p_1}{{\bar p}_1}\phi,  \qquad 
s^*\sigma = \frac{i}{\lambda}e ,
\eee
and combining this to pull-back the vortex equations on $S^3$ \eqref{vortexs3},
 we obtain 
\begin{align}
&(d(s^*\Phi) +i (s^*A)(s^*\Phi))\wedge e = 0, \nonumber \\
&\qquad ds^*A  = \frac {i} {2\lambda^2} (|\phi |^2-1)\bar e \wedge e ,
\end{align}
or, equivalently, 
\begin{align}
&\left(d\left(\frac{p_1}{{\bar p}_1}\phi\right)   +i \left( -R^*\Gamma +   \Gamma+ i d\ln\frac{p_1}{\bar{p}_1}\right) \left(\frac{p_1}{{\bar p}_1}\phi\right)\right)\wedge e = 0, \nonumber \\
&\qquad -d(R^*\Gamma) + d\Gamma  = -(|\phi |^2-1)\mathcal R .
\end{align}
Writing this in terms of the Popov connection $a$, we conclude that 
\begin{align}
&\left(d \phi    -ia \phi\right)\wedge e = 0, \nonumber \\
&\qquad -f = -(|\phi |^2-1)\mathcal R,
\end{align}
which is equivalent to the Popov equations \eqref{Popov}. 
\end{proof}

\subsection{Geometrical interpretation and singularities}
\label{singsect}
It is implicit in our summary, particularly in equation \eqref{popoveq}, that Popov vortices can also be interpreted purely geometrically. A metric viewpoint was emphasised and discussed in the more general context of vortex equation on a Riemann surface with a K\"ahler metric in \cite{Baptista}. In that paper, Baptista pointed out that vortices on a surface with metric $g$ define a new geometry by rescaling with the Higgs field 
\bee
g\rightarrow  g'=|\phi|^2 g.
\eee
The new metric degenerates precisely at the zeros $Z_j$, $j=1,\ldots,2n-2$ of the Higgs field (not necessarily distinct), but its Levi-Civita connection has a Riemann curvature 2-form  which, as explained in \cite{Baptista}, can naturally be extended to the zeros by including delta-function singularities
\bee
\mathcal{R}'= \mathcal{R} + f -2\pi   \sum_{j=1}^{2n-2} \delta_{Z_j}.
\eee
Geometrically, the rescaled metric $g'$ has a singularity with a surplus angle $2\pi n$ at a zero of multiplicity $n$. 
Such singularities can also be thought of as conical singularities with a `negative deficit' angle, i.e.,  with an excess angle. They resemble a ruffled collar, and  are sometimes called `Elizabethan geometries' in the literature. 

By virtue of  $a$ being a connection on a line bundle of  degree $2n-2$, we know that 
\bee
\int_{S^2} f = 4\pi n - 4\pi.
\eee
When integrating $\mathcal{R}'$ this is cancelled by the delta-function contributions,  and so
\bee
\int_{S^2} \mathcal{R}' = 4\pi,
\eee
assuring  that the usual Gauss-Bonnet formula applies to $\mathcal{R}'$. This should be contrasted with  the pull-back curvature  $R^*\mathcal{R}$ which integrates to $4\pi n$.

In the metric interpretation,  the zeros of the Higgs field lead to  singularities whereas the actual singularities of the Higgs field do not appear to play a special role. To understand the geometric interpretation of the singularities of the Higgs field, we need to consider the frame field defined by it. The complexified frame field  
\bee
\phi e = 2 \lambda \frac{p_2' p_1- p_1' p_2}{|p_1|^2  + |p_2|^2}\frac{\bar{p}_1}{p_1} \, dz
\eee
has singularities at each  the zeros of $p_1$, i.e., at each of the pre-images of the singularity of the frame $e$ under the map $R$.  If $q$ is a zero of $p_1$, the behaviour near $q$ is 
\bee
\phi e \sim A \frac{\bar z-\bar q}{z-q}dz,
\eee
 for some constant $A$. Near $z=\infty$, we use again $\zeta =1/ z$ to write the leading term as  
\bee
\phi e \sim B \frac{dz}{z^2} = -B d\zeta, \quad B \; \text{constant},
\eee
 which is smooth. It follows that the winding number of the frame field is localised at the zeros of $p_1$, with each zero (counted with multiplicity) contributing a winding of $4\pi$.  

This interpretation is gauge dependent. 
Using \eqref{Phipull}, we find that  the frame field
\bee
s^*(\Phi \sigma) =  2i \left(\frac{p_2' p_1- p_1' p_2}{|p_1|^2  + |p_2|^2}\right) \, dz
\eee
has no singularities for finite $z$, but has a singularity at $z=\infty$, where it behaves like
\bee
\label{inftysing}
s^*(\Phi \sigma)\sim C \left(\frac{z}{|z|}\right)^{2n} \frac{dz}{z^2}= - C\left(\frac{\bar \zeta}{|\zeta|}\right)^{2n} d\zeta
\eee
for yet another constant $C$.  In this gauge, the full phase rotation of  $4\pi n$ is concentrated  at $z=\infty$.

Our discussion  shows that any description of the magnetic zero-modes in terms of the Popov vortex fields invariably has singularities since the Popov vortex  is  a section of and a connection on a non-trivial bundle, neither of which permits a globally smooth expression.  This  also  applies to the expressions derived in \cite{AMN2,AMN3}, which, in our terminology,  express the magnetic zero-modes in terms of the modulus and phase of a scalar Popov vortex field (whose modulus obeys a Liouville equation). While one can shift the location of the singularities with gauge transformations, one cannot remove them  on $S^2$.

\subsection{Gauge potentials for Cartan connections}
Cartan connections combine the frame and spin connection into a non-abelian connection.
We now show how the results of the previous section can be expressed in the language of Cartan geometry. We first  exhibit a local gauge potential for a Cartan connection constructed from the frame and connection defined by a Popov vortex, and then show how it is related to the the gauge potential
$\mathcal{A}$ used for describing vortices on the 3-sphere in Theorem~\ref{flatvor}.

\begin{lem}
Combine the frame \eqref{compframe} and spin  connection \eqref{spincon} of the 2-sphere  into the $su(2)$ gauge potential
\bee
\label{S2cartan}
\hat{A}=-\Gamma t_{3}+\frac{i}{2\lambda}e \,t_{-}-\frac{i}{2\lambda}\bar{e}\,t_{+}
\eee
defined on the 2-sphere without the south pole. Then the flatness condition for $\hat A$ 
is equivalent to the structure equation \eqref{structure} and Gauss equation \eqref{Gauss} on a $2$-sphere of radius $\lambda$. Moreover, the flatness of the pull-back $R^{*}\hat{A}$ via the rational map $R$  \eqref{Rfirst} is equivalent to the Popov equations being satisfied by the pair $(\phi,a)$ defined via \eqref{popoveq}.
\end{lem}

In the language of Cartan connections, this lemma says that $\hat{A}$ is a gauge potential for a Cartan connection describing the round 2-sphere and that  $R^{*}\hat A$ is a gauge potential for a Cartan connection describing the deformed geometry defined by the vortex $(\phi,a)$.

\begin{proof}
Calculating the curvature of the connection $\hat A$ gives
\bee
F_{\hat{A}}=d\hat{A}+\frac{1}{2}[\hat{A},\hat{A}]=-(\mathcal{R} -\frac{i}{2\lambda^{2}}e\wedge \bar{e})t_{3}+\frac{i}{2\lambda} (de-i\Gamma \wedge e) t_- -\frac{i}{2\lambda} (d\bar{e}+i\Gamma \wedge \bar{e})t_+,
\eee
from which we can read off that the vanishing of the coefficient of $t_{3}$ is equivalent to the Gauss equation and the coefficients of $t_{\pm}$ vanishing are equivalent to the structure equations. 
Using \eqref{popoveq}, we have 
\bee
\label{Popovcartan}
R^{*}\hat A=-(a+\Gamma)t_{3}+\frac{i}{2\lambda} \phi e \,t_{-}-\frac{i}{2 \lambda}\bar{\phi}\bar{e}\, t_{+},
\eee
 with curvature
\bee
R^* F_{\hat A} =-(da-(|\phi|^{2}-1)\mathcal{R})t_{3}+\frac{i}{2\lambda} (d\phi -ia\phi)\wedge e \,t_- -\frac{i}{2\lambda}(d\bar{\phi}+ia\bar{\phi})\wedge\bar{e}\,t_+.
\eee
This being zero is equivalent to the Popov equations \eqref{Popov} being satisfied.
\end{proof}

The gauge potential $R^*\hat A$ inherits singularities from the singularities of  $\phi e$ discussed earlier. In order to treat this more carefully, we use the notion of  a principal divisor $D=\sum n_j q_j$ of degree $n$ on $S^2$. Given such a divisor we construct a bundle over $S^2\setminus \{q_j\}$ by removing the union of the fibres over the $q_j$ from $SU(2)$, obtaining the total space
\bee
\label{PDdef}
P_D= SU(2) \setminus \bigcup_j \pi^{-1} (q_j).
\eee
For a homogeneous polynomial $P$ of degree $n$  in $z_1,z_2$, let $D$ be the divisor of zeros of the associated inhomogeneous polynomial $p$ (so $P(z_1,z_2)= z_1^n p\left(\frac{z_2}{z_1}\right)$).  
Then  we can define the map
\bee
\label{rmap}
r_{P}:P_D  \rightarrow SU(2), \quad r_P = \begin{pmatrix}
\frac{\bar{P}}{|P|}&0\\
0&\frac{P}{|P|}
\end{pmatrix},
\eee
and the pull-back
\bee
r_p=s^*r_P: S^2 \setminus \{q_j\} \rightarrow SU(2).
\eee
It  has the form
\bee
r_{p}=\begin{pmatrix}
\frac{\bar{p}}{|p|}&0\\
0&\frac{p}{|p|}
\end{pmatrix}.
\eee
 For later use we note the behaviour of this matrix under fibre rotations.  
Identifying $h$ with $(z_{1},z_{2})$ as in \eqref{complexp}, we have 
\bee
 \label{requivariance}
r_{P}(he^{\frac{\gamma}{n} t_{3}})=e^{-\gamma t_{3}}r_{P}(h), \quad \gamma \in \left[0,4\pi \right) .
\eee
\begin{lem}
With $s$ defined as in \eqref{locsec}, the gauge potential for the Cartan connection of the 2-sphere is trivialised by $s$:
\bee
\label{Asformula}
\hat A = s^{-1} ds.
\eee
Moreover, if $U$ is the bundle map \eqref{trivialisation} covering the rational map $R=p_2/p_1$,   the gauge potential 
 $R^{*}A$  for the deformed Cartan geometry and the pull-back  via $s$ of  $\A=U^{-1}dU$  are related through the singular gauge transformation $r_{p_1}$:
\bee
R^*\hat A =r_{p_1}^{-1}s^*\left(\mathcal A\right)r_{p_1}+r_{p_1}^{-1}dr_{p_1}.
\eee
\end{lem}

\begin{proof}  The formula \eqref{Asformula} follows by an elementary calculation and comparison with the definition of $e$ and $\Gamma$ in terms of $z$ in \eqref{compframe} and \eqref{spincon}.
With the map $U:S^3\rightarrow S^3$ defined in terms of polynomials $P_1,P_2$ as in \eqref{trivialisation}, and  the map $R:S^2\rightarrow S^2$ defined as in \eqref{Rdef}, 
one checks that, 
\bee
\label{Usexp}
 U\circ s= \frac{1}{\sqrt{|p_1|^2 + |p_2|^2}} \begin{pmatrix}  p_1 & -\bar{p}_2 \\ p_2 & \phantom{-} \bar{p}_1 \end{pmatrix},
\eee
and so, 
choosing  the polynomial $P_1$ used in the definition of $U$ \eqref{trivialisation},
\bee
s\circ R=  (U\circ s)r_{p_1}.
\eee
It follows that 
\bee
 (s\circ R)^{-1} d(s\circ R)= r_{p_1}^{-1}s^*\left(U^{-1}dU\right)r_{p_1}+r_{p_1}^{-1}dr_{p_1}.
\eee
Since $\mathcal A= U^{-1} dU$ and  
\bee
\label{nonabpopov}
R^*\hat A = (s\circ R)^{-1} d(s\circ R),
\eee
the claim follows. \end{proof}

While the 1-form $\mathcal A= U^{-1} dU$  is manifestly smooth on $S^3$,  its pull-back with $s$ is not. 
The map $s\circ U$  \eqref{Usexp}  has a singularity of the form $z^n/|z|^n$ at $z=\infty$, as one would expect since the pull-backs $s^*P_1$ and $s^*P_2$ are local expression for sections of line bundles of degree $n$ over $S^2$ \cite{JS}. It follows that the pull-back $s^* \mathcal A$ is singular at $z=\infty$, with the singularity  already exhibited at \eqref{inftysing}.

\subsection{Cartan geometry}
Our description  of the geometry of the 2-sphere and its pull-back via the rational map $R$  in terms of $su(2)$ gauge potentials has  been  entirely local so far. It is time to address the global geometrical structure behind these gauge potentials. We  will specify the bundles and  the connections  for which \eqref{S2cartan} and \eqref{Popovcartan}  are  local gauge potentials in the language of Cartan geometry, but 
refer the reader to the textbook \cite{Sharpe} and particularly to the PhD thesis \cite{Wise} for general definitions  and facts about Cartan geometry. 

Cartan connections describe the geometry of  manifolds modelled on homogeneous spaces $G/H$ in terms of a connection on a  principal $G$-bundle  $Q$ over this manifold. In order to recover the geometry of a  manifold  from a Cartan connection one needs an additional structure, namely a section of an associated $G/H$ bundle which is transverse to the connection or, equivalently (as explained in \cite{Wise}),  a principal $H$ subbundle $P$ of $Q$ which  is  transverse to the connection. 

Here, we are  interested in the case $G=SU(2),H=U(1)$ and $G/H=S^2$, and only consider flat Cartan connections.
However, we will need to extend the usual framework of Cartan connections to deal with singularities.   Consider a divisor $D$ of degree $n$ on $S^2$, and define   the quotient
\bee
P_{D,n}= P_D/\Z_n,
\eee
where  $P_D$ is as defined in \eqref{PDdef} and we think of $\Z_n$ as the subgroup generated by $
e^{\frac{4\pi}{n} t_3}$, 
acting from the right on $SU(2)$. This is a $U(1)$ bundle over $S^2\setminus\{q_j\}$ with the projection provided by the usual Hopf map  $\pi$ \eqref{Hopf}. It is  a Lens space with $n$ circles removed.

In order to construct the required principal $SU(2)$ bundle,  we define the $SU(2)$-bundle associated to $P_D$ via a $U(1)$ action on $SU(2)$:
\bee
Q_D= \{ (h,g) \in P_D \times SU(2)\}/\sim,
\eee
where $\sim$ is the equivalence relation
\bee
(h,g) \sim \left(he^{\frac{\gamma}{n} t_3},  g e^{\gamma t_3}\right), \quad \gamma \in [0,4\pi).
\eee
 This is an $SU(2)$ bundle over $S^2\setminus \{q_j\}$ with  projection
\bee
\Pi: Q_D\rightarrow S^2 \setminus \{q_j\}, \quad \Pi((h,g)) = \pi (h).
\eee
To make this into  a principal $SU(2)$ bundle, we pick  a homogeneous polynomial $P$ of degree $n$ and consider the map $r_P$  in \eqref{rmap}.
 Then  
 \bee
 \label{gtilde}
\tilde g =g r_P
\eee
  is well-defined on $Q_D$, by \eqref{requivariance}, and we define  the  $SU(2)$  right-action as 
\bee
u: (h,g) \mapsto (h,gr_P u).
\eee

Sections are constructed from  maps satisfying an equivariance condition 
\begin{equation}
\label{equiv}
U:P_D\rightarrow SU(2), \quad  U\left(he^{\frac{\gamma}{n} t_3}\right) =U(h)\left( e^{\gamma t_3}\right), \quad \gamma \in [0,4\pi). 
\end{equation}
This ensures that, for any $U$ satisfying this condition,  
\bee
\label{locS}
S_U:S^2\setminus \{q_j\} \rightarrow Q_D, \quad n \mapsto (h, U(h)), \qquad \pi (h) =n,
\eee
is a well-defined section. 

The Maurer-Cartan form $g^{-1} dg $ 
is not well-defined on $Q_D$ since it is not right-invariant. However, for any $P$, the element $\tilde g$ defined in \eqref{gtilde} is, and so
\bee
\omega =  \tilde g^{-1} d\tilde g 
\eee
is well-defined and  satisfies the equivariance condition for a connection 1-form. It is the  Cartan connection which we are looking for.  The map  $U$ in \eqref{trivialisation}  satisfies \eqref{equiv}. Picking $P=P_1$ and  pulling back $\omega$ to our stereographic coordinate chart via $S_U$
 leads  to the gauge potential 
  \bee
(S_U)^*\omega  = s^*( (U r_{P_1})^{-1} d(U r_{P_1}) = (s\circ R)^{-1} d(s\circ R),
\eee
which, according to \eqref{nonabpopov} and  \eqref{Popovcartan}, indeed  captures the geometry  induced by the Popov vortices.

To end this section,  we also exhibit the second way in which one recovers geometry  from a Cartan connection. As mentioned earlier, this requires a transverse section of the $SU(2)/U(1) \simeq S^2$ bundle associated to the principal $SU(2)$ bundle $Q_D$.  In the trivialisation via \eqref{locS}, this section (often called the Higgs field in the physics literature on Cartan connections) is simply the  constant map 
\bee
\varphi: \C \rightarrow S^2, \quad  z\mapsto t_3,
\eee
where we think of $t_3$ as an element of unit sphere inside $su(2)$. 
The geometry is recovered from the covariant derivative 
\bee
D_{R^*\hat A} \varphi = [R^*\hat A, t_3],
\eee
which extracts the frame $\phi e$ from the gauge potential $R^*\hat A$. 
If we apply the gauge transformation $s^*(Ur_{P_1})$, the gauge potential vanishes in the new gauge,  but the transverse section is now 
\bee
\tilde \varphi : z\mapsto Ut_3U^{-1},
\eee
which, up to stereographic projection, is our rational map $R$. In other words, the rational map which solves the Popov vortex equations is the `transverse Higgs field' of Cartan geometry in a particular gauge.  The geometry is still recovered via the covariant derivative, but since the gauge potential vanishes now, this is simply the exterior derivative $d \tilde  \varphi$ which, modulo stereographic projection, indeed reproduces the formula for the frame $\phi e$ in terms of the derivative of $R$.

\section{Summary and Outlook}
 The  equations, spaces and maps studied in this paper are summarised in Fig.~\ref{summary}, with magnetic zero-modes on the  top left  of the picture and the Popov vortex equations on the bottom. The geometry of the 3-sphere, as encoded in the Maurer-Cartan form $h^{-1}dh$,  and its pull-back via the bundle map $U$ provides a  unifying point of view for  both, and leads to the explicit and smooth description  of both  vortex  zero-modes and of vortices.
 
  We believe that our results provide a fully geometrical understanding of magnetic zero-modes on $\R^3$, and would like to stress the practical advantage of having  manifestly singularity-free and square-integrable formulae. Previous  expressions based on the Loss-Yau formula \eqref{LossYauformula}  have singularities where the spinor field vanishes, see also our discussion  at the end of  Sect.~\ref{singsect}.

 \begin{figure}[!htbp]
 \centering
\includegraphics[width=14truecm]{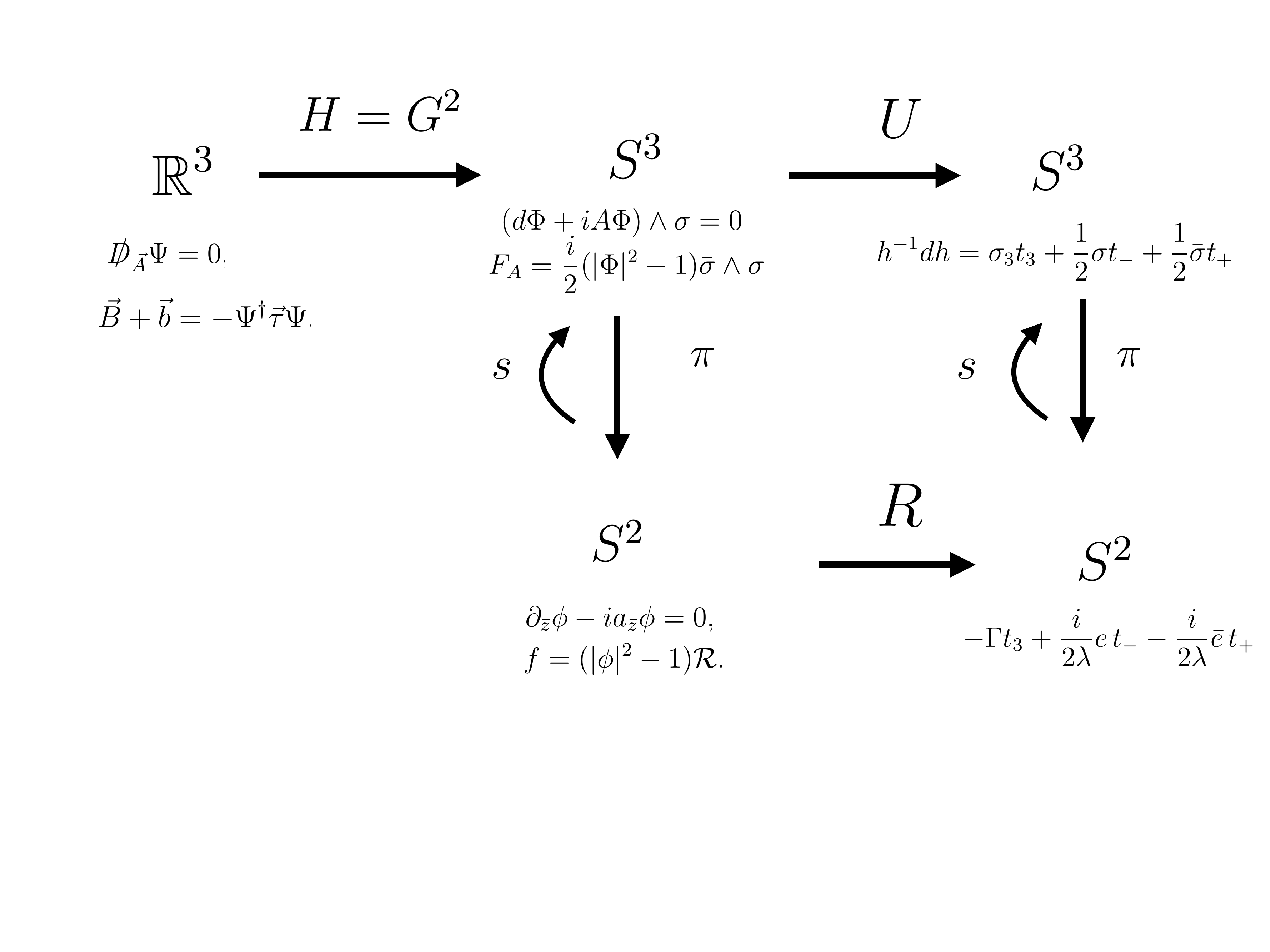}
 \vspace{-2cm}
 \caption{A summary  of the equations and maps studied in this paper}
 \label{summary}
\end{figure}
 %\vspace{3cm}

The diagram in Fig.~\ref{summary} shows that the structures studied in this paper are closely related to many of the  most studied topological solitons \cite{MS}.  Apart from  the obvious vortex configurations, the maps $U$ and their pull-backs $U\circ H$ are topologically  Skyrme fields on $S^3$ and $\R^3$. The rational maps $R$  can also be interpreted  as  `lumps' or baby Skyrmions on the 2-sphere, while the composite maps  $H\circ U\circ\pi= H\circ \pi\circ R$ are  topologically  Hopfions. As discussed in Sect.~\ref{magfields},  the magnetic fields on $\R^3$ obtained by pulling back the area form on the unit 2-sphere via these maps are  examples of  linked magnetic fields of the kind studied by Ra\~{n}ada and  more recently, in more detail and with  many pictures, in \cite{IB}. Finally, the equation obeyed by vortex zero-modes on $\R^3$ is related to the Seiberg-Witten equations on $\R^4$  with a sign flipped \cite{AMN1}.

It is clear that much of what we discussed in this paper  has a close analogue in a  Lorentzian and hyperbolic setting, where the 2-sphere is replaced by hyperbolic 2-space, the 3-sphere by 3-dimensional Anti de-Sitter space (or $SU(2)$ by $SU(1,1)$), and Euclidean $\R^3$ by Minkowski space $\R^{2,1}$. Spelling this out is the topic of a forthcoming paper \cite{RS2}. However, one should also consider further generalisations suggested by the origin of the Popov equations in self-duality.

As we briefly mentioned in the Introduction,  the Popov  equations are symmetry reductions of the self-duality equations  for $SU(1,1)$ instantons  on $\R^4$.  In fact, there is a whole family of integrable vortex equations recently studied by Manton \cite{Manton2} which have  similar links to self-duality equations, with the relevant gauge group depending on the vortex equation \cite{Dunajksi}. It is intriguing that, for Popov vortices,  the non-abelian gauge group $SU(1,1)$  needed for the self-dual  connection differs from the $SU(2)$ we used in our description in terms of flat Cartan connections. It would be interesting to understand both viewpoints  and their relationship systematically for the family of integrable vortex equations studied in \cite{Manton2}.

To end, we point out that  interactions of spinors with linked magnetic fields of the form \eqref{Ranadafield}
are currently  much discussed in atomic and  condensed matter physics, where the spinors  arise as an effective description of  nearly degenerate states of ultra-cold atomic Bose-Einstein condensates,  and the magnetic field as the curvature of a Berry connection, see, e.g., the papers \cite{JRO} and  \cite{DGJO} for a review. Our explicit expression for more general magnetic zero-modes may prove useful in that context. 

\vspace{0.5cm}

\noindent {\bf Acknowledgements} \, CR acknowledges an EPSRC-funded PhD studentship. We thank Patrik \"Ohberg for discussions about possible applications of our ideas.


\begin{thebibliography}{100}
\bibitem{LY}
M.~{Loss} and H.-T. {Yau}.
\newblock {Stability of Coulomb systems with magnetic fields: III. Zero energy
  bound states of the Pauli operator}.
\newblock {\em Communications in Mathematical Physics}, 104:283--290, 1986.
  
\bibitem{AMN1}
C.~{Adam}, B.~{Muratori}, and C.~{Nash}.
\newblock {Non-L$^{2}$ solutions to the Seiberg-Witten equations}.
\newblock {\em Journal of Mathematical Physics}, 41:5875--5882, 2000.

\bibitem{AMN2}
C.~{Adam}, B.~{Muratori}, and C.~{Nash}.
\newblock {Hopf instantons and the Liouville equation in target space}.
\newblock {\em Physics Letters B}, 479:329--335, April 2000.

\bibitem{AMN3}
C.~{Adam}, B.~{Muratori}, and C.~{Nash}.
\newblock {Hopf instantons in Chern-Simons theory}.
\newblock {\em Phys. Rev. D}, 61(10):105018, 2000.

\bibitem{ES}
L.~{Erd{\"o}s} and J.~P. {Solovej}.
\newblock {The kernel of Dirac operators on S$^{3}$ and R$^{3}$}.
\newblock {\em Reviews in Mathematical Physics}, 13:1247--1280, 2001.

\bibitem{DM} G.~Dunne and H.~Min. 
\newblock{ Abelian zero-modes in odd dimensions}.
\newblock{ \em Phys. Rev. D} 78: 067701, 2008.

\bibitem{Min} H.~Min
\newblock{Fermion zero-modes in odd dimensions}.
\newblock{\em J. Phys. A: Math. Theor.} 43: 095402, 2010.

\bibitem{Popov}
A.~D. {Popov}.
\newblock {Integrable vortex-type equations on the two-sphere}.
\newblock {\em Phys. Rev. D}, 86(10):105044, 2012.

\bibitem{Manton1}
N.~S. {Manton}.
\newblock {Vortex solutions of the Popov equations}.
\newblock {\em Journal of Physics A Mathematical General}, 46(14):145402, 2013.

\bibitem{Baptista}
J~M. ~{Baptista}.
\newblock {Vortices as degenerate metrics}.
\newblock {\em Letters in Mathematical Physics}, 104(6):731--747, 2014.

\bibitem{JS}
R.~{Jante} and B.~J. {Schroers}.
\newblock {Dirac operators on the Taub-NUT space, monopoles and SU(2)
  representations}.
\newblock {\em Journal of High Energy Physics}, 1:114, 2014.

\bibitem{JS2}
R.~{Jante} and B.~J. {Schroers}.
\newblock {Spectral properties of Schwarzschild instantons}.
\newblock {\em Classical and Quantum Gravity} 33: 205008, 2016.


\bibitem{Ranada1}
A.~F.~Ra\~{n}ada.
\newblock Knotted solutions of the Maxwell equations in vacuum.
\newblock {\em Journal of Physics A: Mathematical and General}, 23(16):L815,
  1990.

\bibitem{IB}
William T.~M. Irvine and Dirk Bouwmeester.
\newblock {Linked and knotted beams of light}.
\newblock {\em Nature Physics}, 4(9):716--720, 2008.


\bibitem{Ranada2}
A.~F.~Ra\~{n}ada.
\newblock On the magnetic helicity.
\newblock {\em European Journal of Physics}, 13(2):70, 1992.


\bibitem{Sharpe}
R. ~W. ~Sharpe.
\newblock{Differential Geometry: Cartan's generalization of Klein's Erlangen
program.}
\newblock Springer, New York, 1997.

\bibitem{Wise}
D.~K.~Wise.
\newblock{Topological gauge theory, Cartan geometry and gravity.}
\newblock PhD Thesis, University of California (Riverside), 2007, p 43.

\bibitem{MS}
N.~S. Manton and P.~M. Sutcliffe.
\newblock {\em Topological solitons.}
\newblock Cambridge monographs on mathematical physics. Cambridge University
  Press, Cambridge, 2004.
  
\bibitem{RS2}
C.~Ross and B.~J. {Schroers}.
\newblock{}
\newblock{In preparation}

\bibitem{Manton2}
N.~S. {Manton}.
\newblock {Five vortex equations}.
\newblock {\em Journal of Physics A: Mathematical and General}, 50(12):125403, 2017.

\bibitem{Dunajksi}
F.~Contatto and M.~Dunajski.
\newblock {Manton's five vortex equations from self-duality}.
\newblock arXiv:1704.05875.
  

\bibitem{JRO}
G.~Juzeliunas , J.~Ruseckas  and P. ~\"Ohberg.
\newblock{Effective magnetic fields induced by EIT in
ultra-cold atomic gases}.
\newblock{\em Journal of Physics B: Atomic,  Molecular and Optical  Physics}, 38:4171-4183, 2005.

\bibitem{DGJO}
  J.~Dalibard, F.~Gerbier, G.~Juzeliunas and P.~\"Ohberg.
  \newblock{Colloquium: Artificial gauge potentials for neutral atoms}.
  \newblock{\em Reviews of Modern Physics}, 83:1523, 2011.
  
  
\end{thebibliography}
\end{document}